\documentclass[lettersize,journal]{IEEEtran}
\def \D {\mathcal{D}}
\def\expect{\mathbb{E}}
\def \eqd {\overset{\Delta}{=}}
\def \ind{{\bf 1}}
\def \policy{\pi}
\def \policyvec{\boldsymbol{\pi}}
\def \policyset{\Pi}
\def \P{\mathbb{P}}
\def \btau{\boldsymbol{\tau}}
\def \bomega{\boldsymbol{\omega}}
\def \bq{\mathbf{q}}
\def \bq{\mathbf{q}}
\def \qq{\mathbf{q}}
\def \N{\mathbb{N}}
\def \trco {\text{rco}}
\def \tco {\text{co}}
\def \tcoe {E}
\def \tcot {T}
\def \too {\text{o}}
\def \toe {{E-O}}

\def \tc {\text{c}}

\def \um {{(m)}}
\usepackage{amsmath,amsthm,amssymb,amsfonts}
\usepackage{algorithmic}
\usepackage{algorithm}
\usepackage{array}
\usepackage{subcaption}
\usepackage{textcomp}
\usepackage{stfloats}
\usepackage{url}
\usepackage{xcolor}
\usepackage{verbatim}
\usepackage{graphicx}
\usepackage{cite}
\newtheorem{lemma}{Lemma}
\newtheorem{theorem}{Theorem}

\begin{document}

\title{Optimal Edge Caching For Individualized Demand Dynamics}

\author{
  Guocong~Quan,
  Atilla~Eryilmaz,~\IEEEmembership{Senior~Member,~IEEE},
  Ness~Shroff,~\IEEEmembership{Fellow,~IEEE}
  \thanks{The authors are with The Ohio State University, Columbus, OH 43210, USA 
(emails: \{quan.72, eryilmaz.2, shroff.11\}@osu.edu).}
}



\maketitle

\begin{abstract} 
The ever-growing end user data demands, 
and the simultaneous reductions in memory costs are fueling edge-caching deployments.
Caching at the edge is substantially different from that at the core 
and needs to take into account the nature of individual data demands. 
For example, an individual user may not be interested in 
requesting the same data item again, if it has recently requested it. 
Such individual dynamics are not apparent in the aggregated data requests at the core 
and have not been considered in popularity-driven caching designs for the core.
Hence, these traditional caching policies could induce significant inefficiencies 
when applied at the edges. 
To address this issue, 
we develop new edge caching policies optimized for the individual demands 
that also leverage overhearing opportunities at the wireless edge. 
With the objective of maximizing the hit ratio, 
the proposed policies will actively evict the data items 
that are not likely to be requested in the near future,
and strategically bring them back into the cache through overhearing 
when they are likely to be popular again. 
Both theoretical analysis and numerical simulations demonstrate that 
the proposed edge caching policies could outperform the popularity-driven policies 
that are optimal at the core.
\end{abstract}

\begin{IEEEkeywords}
    edge caching, demand dynamics, overhearing
\end{IEEEkeywords}

\section{Introduction}
\label{sec:intro}
Data demands are growing exponentially,  
driven by the rapid proliferation of edge devices such as the Internet of Things (IoT), 
and increasingly capable hand-held devices. 
At the same time, memory is becoming increasingly cheaper, larger, and faster.  
These two forces are creating an ideal environment for the large-scale deployment of edge caching 
to support fast data retrieval~\cite{qian2012web, 
niyato2016novel, wang2014cache, zhang2018cooperative}.
While, extensive 
studies~\cite{jia2017efficient,li2016popularity, 
huang2013analysis, cidon2015dynacache} 
have been conducted 
to optimize caching strategies 
for relatively stationary data demands at the network core, 
caching at the edges due its individualized demand dynamics, 
is quite different from the core, 
and therefore should be studied in their own right.  
In this paper, we will propose new caching policies 
optimized for the individualized data demands at the wireless edges. 

\subsection{Challenge: Individualized Demands at Network Edges}
At the network core, 
data demands are aggregated from 
a large number of end-users, as shown in Fig.~\ref{fig:core_caching}.
Thus, the demand dynamics of each individual user 
could become negligible, which leads to 
relatively stationary data popularities for the population.
A number of popularity-driven policies have been 
proposed for optimizing caching 
at the core~\cite{li2016popularity,jiang2002lirs}. 
Inspired by the observation that data items recently requested by one user 
are very likely to be requested again by others with similar interests, 
the least recently used (LRU) policy
estimates the popularity by the data recency 
and caches the most recently requested data items. 
The LRU policy and its variants have been widely implemented at the core,
and validated to achieve considerably good 
performance~\cite{nishtala2013scaling,atikoglu2012workload,jelenkovic1999asymptotic}.
\begin{figure}[h]
    \centering
    \vspace{-3mm}
    \includegraphics[width=8.6cm]{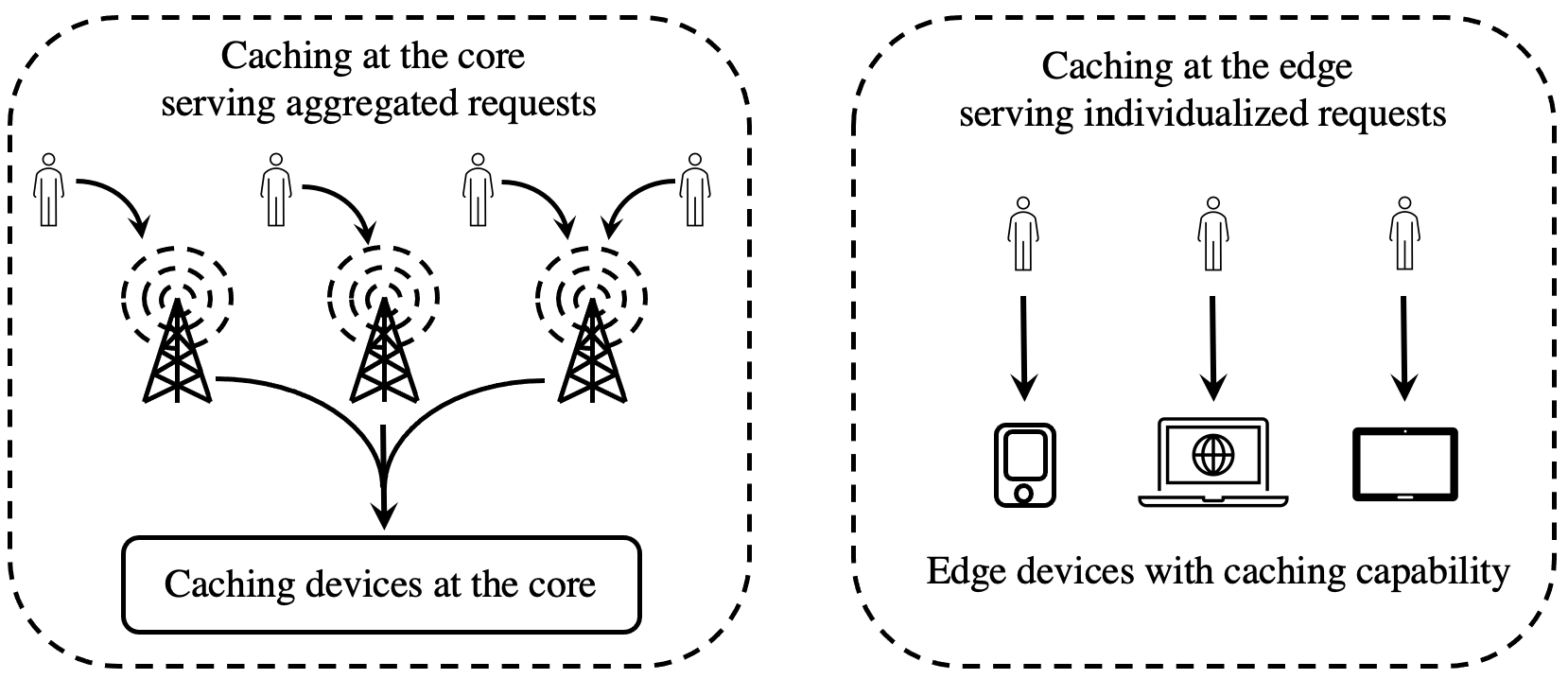}
    \vspace{-1mm}
    \caption{Caching at the core v.s. caching at the edge.
    }
    \vspace{-0mm}
    \label{fig:core_caching}
\end{figure}

In contrast, edge caching serves a small group of users, 
or even a single user (e.g., caching in smartphones),
where the data demands are more individualized.
Those have fundamentally different dynamics than the population demand models.
In particular,
\textit{after requesting a data item, 
the user may not be likely to request the same data item again
in the near future}.
One supportive reason is that users may lose interest in seeing the same or 
similar content repeatedly. 
A common methodology applied by recommendation systems is 
to avoid presenting the same or similar content consecutively~\cite{adomavicius2011improving, 
hurley2013personalised,qian2017image}. 
Another evidence is that users may need some time to process the recently requested data. 
A trace analysis on Yelp
validates that there are considerable time gaps between users' actions~\cite{tadrous2018action}.
Hence, 
the popularity of a data item at the edges could change dramatically even
after every request for it. 
The popularity-driven policies designed for the core
are not able to make adaptive caching decisions for such individualized demands, 
and therefore may achieve poor performance at the edges.

In Fig.~\ref{fig:LRU}, we illustrate a motivating example showing 
that the LRU policy could make irrational decisions for edge caching, 
since the recently requested data items typically have small popularities at the edges,
which is completely opposite to the experience at the core.
Assume that an individual user will not request the same
data item in the near future after each 
request\footnote{The data requests are generated with $N=1000, b = 50,
s_i=5000, \beta_i=c\cdot i^{-1.4}, c = 1/\sum_{j=1}^N\beta_j = 0.3392,1\leq i\leq N$, 
where the detailed parameter definitions are introduced in Section~\ref{sec:demand_dynamics}}. 
We simulate the hit ratio achieved by a single LRU cache serving aggregated data requests from a group of users. 
The LRU policy achieves good performance when the number of users is large 
(i.e., the network core scenario). 
However, the performance degenerates significantly as the number of users decreases 
(i.e., the network edge scenario). 
Interestingly, when the cache serves one user, the hit ratio will decrease to zero,
which indicates that the LRU policy almost always makes the wrong decisions.
\begin{figure}[t]
    \centering
    \includegraphics[width=7cm]{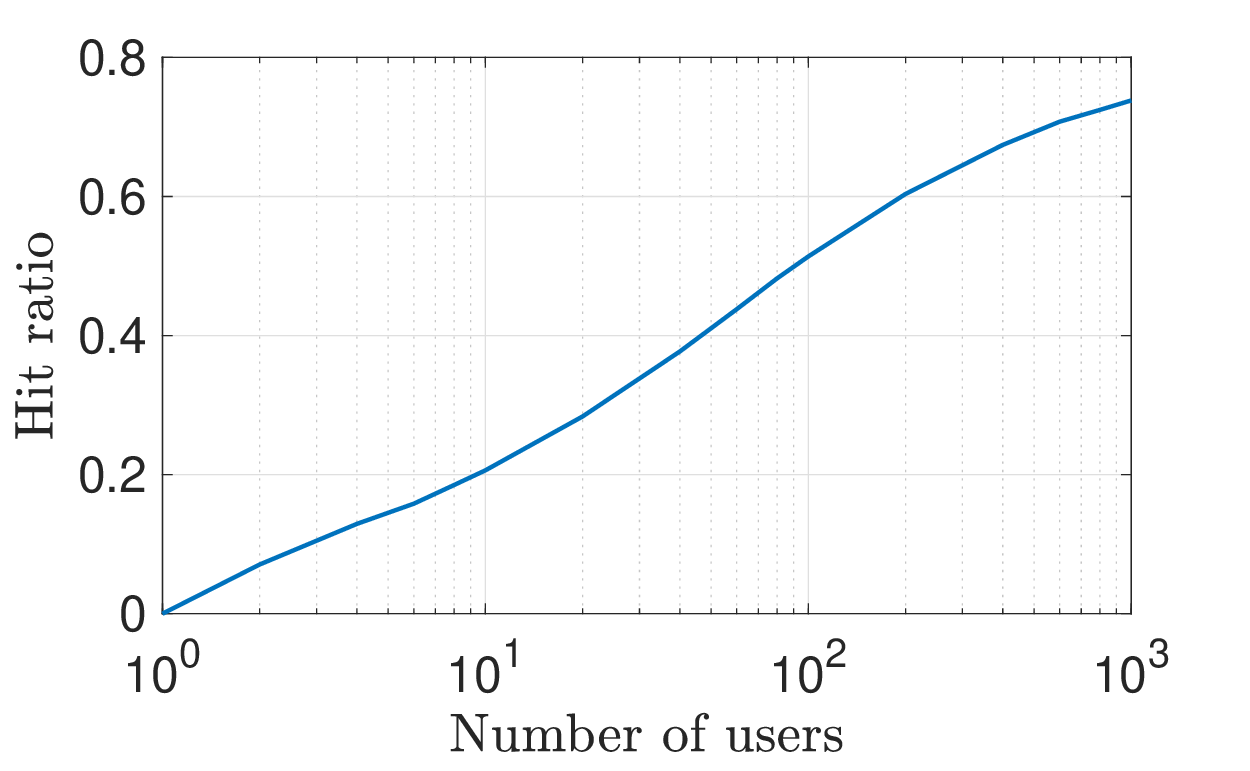}
    \vspace{0mm}
    \caption{Degenerate performance of LRU when serving 
    a small group of users.}
    \vspace{0mm}
    \label{fig:LRU}
\end{figure}

\subsection{Solution: Active Eviction and Strategical Overhearing}
To address this issue, we develop new adaptive edge caching policies 
customized for the individualized demands.
In an ideal case, the policy should frequently update the cache content 
and only store data items 
that are most likely to be requested in the near future. 
We will leverage the overhearing opportunities 
at the wireless edges to mimic this ideal design.  

Specifically, an edge cache can overhear the broadcasted data items
over the wireless channels, 
even when it is not the intended receiver. 
To achieve high caching efficiency,  
we may actively evict the recently requested data items 
that will not be needed in the near future,
and strategically bring them back into the cache later through overhearing 
when their popularities rise up sufficiently again. 
With the objective to maximize the overall hit ratio, 
we optimize the eviction and overhearing decisions 
for two different settings depending on 
how the overhearing opportunities are generated. 
Under the time-driven overhearing setting (cf. Section~\ref{sec:single}),
the overhearing opportunities are described 
by Poisson processes that are independent of the data requests
and out of the designer's control.
Under the event-driven overhearing setting (cf. Section~\ref{sec:multi}),
the overhearing opportunities are generated 
when an item is requested and is unavailable at a user, 
which triggers its broadcast over the wireless channel,
hereby generating an overhearing opportunity for all other users.  
Our contributions are summarized as follows.
\begin{itemize}
    \item With the objective to maximize the overall hit ratio, 
    we propose an optimal caching and overhearing policy 
    for the time-driven overhearing setting. 
    Specifically, we first prove that the hit ratio maximization problem 
    is nonconvex.
    By exploiting an informative structure of the optimal solutions,
    we then convert the nonconvex problem to a convex one and propose
    efficient algorithms to solve it
    (see Section~\ref{sec:single}). 
    \item We propose an asymptotically optimal caching and overhearing policy
    for the event-driven overhearing setting. 
    Although the overhearing process is not fully tractable
    under this setting, 
    we are inspired by the structure of the optimal policy under time-driven overhearing 
    and propose a policy for event-driven overhearing, which is asymptotically optimal 
    when the number of edge caches in the system is sufficiently large
    (see Section~\ref{sec:multi}). 
    \item We extend our main results for both time-driven and event-driven overhearing scenarios 
    to a more general data demand setting, where different users can have heterogeneous demand patterns
    (see Section~\ref{sec:gen}). 
    \item We conduct extensive simulations to validate that 
    the proposed policies can achieve better performance than a few benchmarks
    (see Section~\ref{sec:eva}). 
\end{itemize}

\subsection{Related Works}

Conventional caching analysis for stationary data demands typically assumes
an independent reference model (IRM), where the data requests are
assumed to be generated from a stationary popularity distribution 
independently. 
Popularity driven caching policies are proposed for such scenarios
in different systems~\cite{li2016popularity,cho2012wave}.
Historical request information 
including data recency and frequency
are commonly leveraged to estimate the popularity,
and inspire the design of LRU, LFU, LIRS and other 
variants~\cite{jiang2002lirs,o1993lru,friedlander2019generalization, 
beckmann2018lhd,quan2019new}.
Among the various caching policies, 
the time-to-live (TTL) based policies have garnered significant attention, 
since they are not only easy to implement in real practice, 
but also provide tractability and flexibility to optimize different system 
goals~\cite{domingues2021role,ferragut2016optimizing,dehghan2019utility,jung2003modeling}. 
However, how to design a good TTL-based policy for edge caching with individualized 
demand dynamics still lacks a systematic study.

To characterize the non-stationary data demands whose popularities
may evolve over time, 
a shot noise model (SNM) is proposed in~\cite{traverso2013temporal},
where the request process of a data item is described by a
time-inhomogeneous Poisson process.
Compared to IRM,  SNM could better characterize the temporal 
locality and is validated by real data traces collected from 
more than 10000 IPs.
However, under the general SNM, the theoretical analysis of 
some caching strategies may become intractable. 
To address this issue, an ON-OFF traffic model is proposed 
in~\cite{garetto2015efficient},
which captures time-variant data popularities and 
supports efficient analysis for a number caching strategies. 
An age-based threshold (ABT) caching policy is proposed 
for small user populations 
under SNM~\cite{leconte2016placing}. 

{
\color{black}
Numerous studies have explored methods for tracking dynamic data demands and optimizing
caching decisions to achieve better efficiency~\cite{hoang2018dynamic, qi2019learning, kumar2020optimized,
gao2020design, zong2022cocktail}.
However, they are different from this paper in the following aspects.
1) Existing works typically consider the aggregated demands from a group of users and 
attempt to track the dynamic demands over a relatively long time period
by collecting historical requests from different users.
The individualized demands that could change dramatically
in a short time period (e.g., after each request for the data item)
have not been well addressed by existing works. 
2) These works did not explore the joint optimization of 
overhearing and caching decisions when demands are dynamically changing. 
In this paper, we fill the gap by proposing a new model to describe 
the individualized demand dynamics and designing new edge caching policies 
that achieve provably better performance by leveraging overhearing opportunities
at the wireless edges.

Another prevalent category of dynamics in caching problems is the content dynamics,
where each data item could be occasionally refreshed, rendering older versions obsolete.
Different caching strategies have been explored to optimize caching efficiency and content freshness
~\cite{abolhassani2021single,bastopcu2020information, zhang2021aoi}.
We note that such content dynamics are different from the demand 
dynamics investigated in this paper, because 
the data popularity changes under content dynamics are triggered by the refresh of content sources, 
while popularities under demand dynamics are changing with users' actions (e.g., recent requests). 
Furthermore, edge caching with overhearing opportunities have been shown great potential to 
improve energy efficiency, transmission delays and data freshness
~\cite{poularakis2016exploiting,cui2016analysis,
zhou2017optimal,amiri2018caching,abolhassani2020delay},
but these designs didn't consider the individualized demand dynamics. 
}

\section{Model Description}
\label{sec:model}
In this section, we will first formally model the edge demand dynamics 
by ON-OFF processes. 
Then we will introduce TTL based caching and overhearing policies 
and formulate a hit ratio maximization problem.

\begin{figure}[h]
    \centering
    \vspace{-2mm}
    \includegraphics[width=8.3cm]{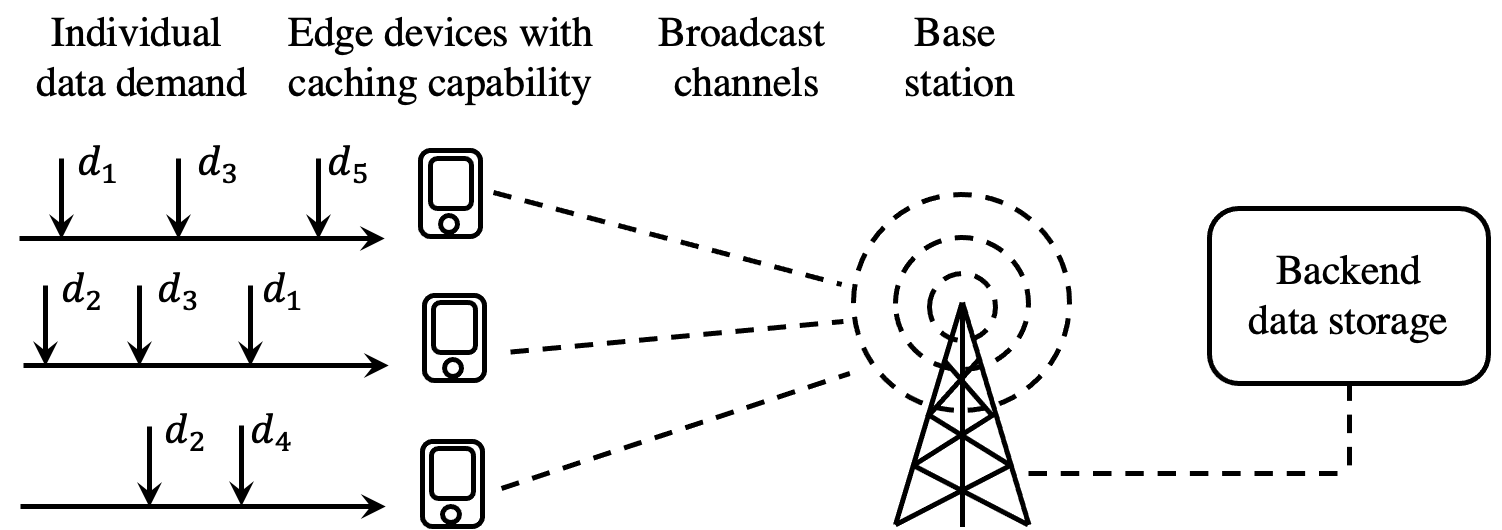}
    \vspace{-1mm}
    \caption{Edge caching with individual data demand.
    }
    \vspace{-0mm}
    \label{fig:model}
\end{figure}

\subsection{Individual Demand Dynamics}
\label{sec:demand_dynamics}
Consider $M$ edge caches connected to a base station through 
wireless channels, as shown in Fig.~\ref{fig:model}.
Each edge cache serves data requests from a single user. 
We use $m$, $1\leq m \leq M$, to index an edge cache or the user served by the edge cache interchangeably. 
Let $\{d_i$, $1\leq i \leq N\}$ denote a set of $N$ 
distinct data items.
Assume that the data items are of unit size and
each edge cache has a size of $b$, $0 < b \leq N$.
Each edge cache serves an individual user independently. 
If the requested data is stored in the cache, 
then the request could be served immediately with a low latency, 
which is called a cache hit.
Otherwise, the requested data has to be obtained from the backend data storage
and sent back to serve the user's demand,
which is called a cache miss.

{\color{black}
To characterize the demand dynamics of individual users, 
we model the requests for the data item $d_i$, $1\leq i\leq N$, 
generated by each user 
as a renewable ON-OFF process. 
Specifically, after the user requests $d_i$,
he/she will not request it again within $s_i$ units of time, 
which is the OFF period. 
The OFF period can effectively capture the transfer of user interests 
as well as the time gaps between users' actions
as demonstrated by a trace analysis on Yelp~\cite{tadrous2018action}.
After the OFF period, the next request for $d_i$
will be generated according to a Poisson process with rate $\beta_i$,
which is the ON period. 
Without loss of generality, we assume that 
the data items are indexed such that $\beta_i$'s are decreasing with 
respect to $i$.
When a new request is generated in the ON period, 
a subsequent OFF period starts immediately and 
the ON-OFF process is renewed. 

For example, in Fig.~\ref{fig:on-off} we illustrate a sequence of requests
for the data item $d_1$ with $s_1 = 2$ and $\beta_1 = 1$.
The first request for $d_1$ is initiated at epoch 0.
After the first request, there is an OFF period of a fixed duration 2,
during which, the user is not interested in requesting $d_1$.
Starting from epoch 2, the first OFF period ends, ushering in an ON period, where 
the user would request $d_1$ again. 
Within the ON period, the next request for $d_1$ will be generated based on a Poisson process with a rate 1. 
In this example, the next request occurs at epoch 4, 
at which time, the second OFF period starts and the whole process is renewed. 
\begin{figure}[h]
    \vspace{-2mm}
    \centering
    \includegraphics[width=7cm]{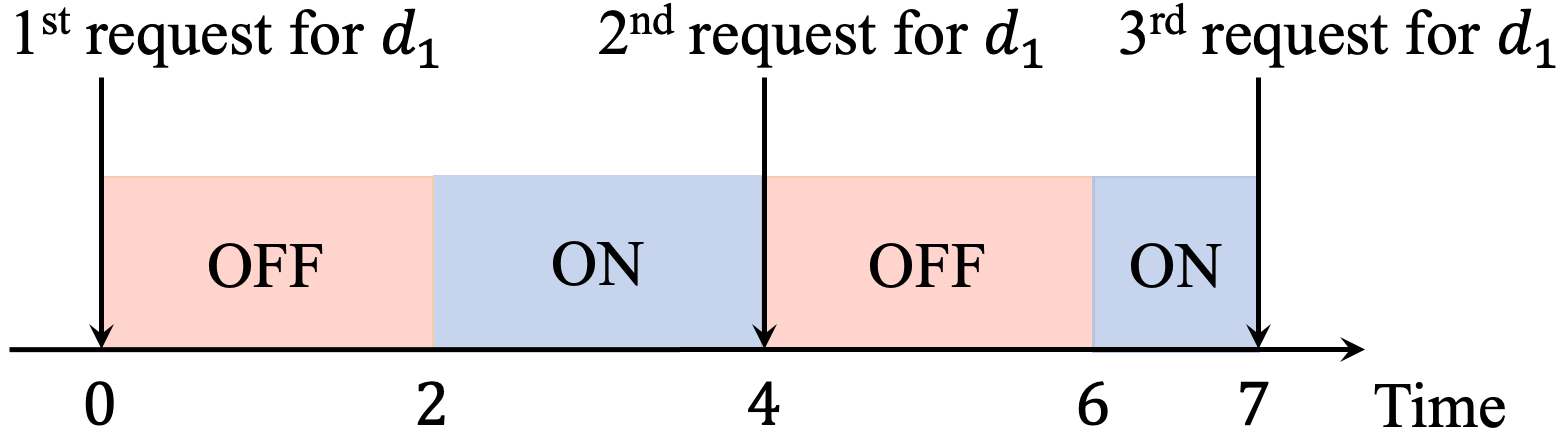}
    \vspace{-2mm}
    \caption{\color{black}Individualized demands characterized by renewable ON-OFF processes.}
    \vspace{-2mm}
    \label{fig:on-off}
\end{figure}

The proposed renewable ON-OFF process describes how a user's demand for a data item will evolve based on the user's recent requests for it. 
We note that the proposed model could characterize different demand patterns depending on the value of $s_i$
\begin{itemize}
    \item When $0 < s_i < +\infty$, the data item won't be requested again in the near future, if the user has recently made a request for it. 
    However, over time, the user might regain interest in it.
    For example, music within a playlist typically follow this demand pattern.
    \item When $s_i = +\infty$, $d_i$ will never be requested again after the first request for it.
    For example, the weather information for today is rarely in demand beyond the current day.
    \item
    When $s_i = 0$, the requests for data item $d_i$ will be generated following a 
    Poisson arrival process with a constant rate $\beta_i$, 
    which indicates that users are consistently interested in such data items.
    This setting corresponds to the conventional accumulated demands at the network core. 
\end{itemize}
In this paper, we assume that the parameters $s_i$ and $\beta_i$ are fixed and known, 
and focus on how we should update the edge cache content 
for a set of candidate data items with different demand patterns (i.e., different $\beta_i$, $s_i$ values). 
In real practice, $s_i$ and $\beta_i$ could be unknown and different approaches could be applied to estimate them.
As an illustration, we can employ clustering algorithms to group users with similar interests, 
enabling us to leverage the observed historical requests from similar users to estimate parameters for others~\cite{hoang2018dynamic}.
}

In the main paper, we consider the homogeneous demand dynamics, 
where different users have the same request pattern 
(i.e., $s_i$, $\beta_i$ only depend on the data item $d_i$ and are identical for different users). 
This setting is particularly relevant to scenarios where individual users 
have common data interests (e.g., on-trend music and TV series).
The homogeneous setting helps us focus on the impact of individualized demands at the network edge
as opposed to the aggregated demands at the network core.
Later, in Section~\ref{sec:gen}, we will show that most of the theorems and insights obtained 
for the homogeneous setting are also valid when users have heterogeneous demands.

\subsection{Overhearing Opportunities}
To improve the edge caching efficiency under such dynamic data demands, 
we will leverage overhearing opportunities over the wireless channels. 
Specifically, the base station can broadcast data items from time to time. 
When the cache overhears a data item, it may decide to store it or not
based on the adopted caching and overhearing policies. 
Since the users are assumed to have common data interests, 
the overheard data item could potentially satisfy the demand of multiple users,
and therefore, improves the caching efficiency. 
Note that the privacy concerns or the data encryption are not considered in this paper.
Developing efficient edge caching policies with privacy protection 
is a crucial avenue for future research.
However, this topic falls outside the scope of this paper.

In this paper, we will investigate optimal caching and overhearing policy
under the following two different overhearing scenarios depending on 
how the overhearing opportunities are generated.

\noindent
\textbf{Scenario 1 (Time-driven overhearing):} 
In this scenario, we assume that the base station will broadcast each data item 
based on independent Poisson processes with given fixed rates. 
The caches could passively overhear and need to decide whether to store the 
overheard data items. 
{\color{black}
Note that the overhearing processes considered in this scenario are fixed and independent of caching decisions.
The time-driven overhearing opportunities are given by the environment
and our goal is to find good policies to leverage these opportunities.
This simple setting could provide us informative insights to 
optimize caching decisions with overhearing opportunities,
which inspires the policy design for more realistic settings in Scenario~2. 
}

\noindent
\textbf{Scenario 2 (Event-driven overhearing):}
In this scenario, we consider a more realistic setting.
When a cache misses some data item, e.g., $d_1$,
the base station will fetch $d_1$ from the backend storage and 
send it back to the cache over the broadcast channel.
Meanwhile, the other caches could overhear $d_1$ and decide whether to 
cache it or not. 
{\color{black}
Unlike time-driven overhearing,
the event-driven overhearing opportunities are not fixed or given by the environment.
They are shaped by the miss behaviors, 
which, in turn, are determined by the caching policies. }
Thus, the policy design for this scenario is expected to be more challenging.

\subsection{TTL-Based Caching and Overhearing Policies}
\label{sec:model_policy}
To achieve as many hits as possible,  
the key question to answer is how to update the cache content. 
In this paper, we consider the design where each data item 
can be updated separately based on its own rule. 
Formally,
we use a vector $\policyvec = (\policy_1, \policy_2, \cdots, \policy_N)$
to denote the policy for managing the $N$ data items, 
where each element $\policy_i$ is the update rule for a data item $d_i$, $1\leq i\leq N$.
To avoid possible confusion, we call $\policyvec$ a policy and each $\policy_i$
an item policy. 

First, let us consider the item policies that belong to 
the following TTL based caching and overhearing item policy set. 
Define
\begin{align}
    \Pi^\tco = \{\pi^\tco(\tau, \omega): \omega \geq \tau\geq 0\},
\end{align}
where an item policy $\pi^\tco(\tau,\omega)$
is determined by two parameters for each data item, 
i.e., the caching TTL $\tau$ and the deaf TTL $\omega$.
The superscript ``co'' stands for caching and overhearing. 

In particular, 
assume that the data item $d_i$ is served by the item policy $\pi^\tco(\tau_i,\omega_i)$.
Then every time the data item $d_i$ is requested,
it will be loaded into the cache regardless of a hit or a miss.
Meanwhile, a caching timer with duration $\tau_i$ 
and a deaf timer with duration $\omega_i$ ($\omega_i\geq \tau_i$) 
will be initiated. 
{\color{black}
In Fig.~\ref{fig:policy}, we illustrate an item policy $\pi^\tco(\tau_1,\omega_1)$
for $d_1$ with $\tau_1 = 3$ and $\omega_1 = 7$.
\begin{enumerate}
    \item Until the caching timer expires,
    $d_i$ will be cached but promptly evicted once the timer runs out.
    In Fig.~\ref{fig:policy},
    $d_1$ will be cached during the period [0, 3] and evicted at epoch 3. 
    \item Before the deaf timer expires, the cache refrains from loading $d_i$ into cache via overhearing.
    In Fig.~\ref{fig:policy}, although $d_1$ is broadcasted at epoch 5, it will not be loaded into cache
    based on the chosen policy.
    \item After the deaf timer expires, the cache will opportunistically store
    $d_i$ via overhearing when it is broadcasted.  
    Once a request for $d_i$ is generated and fulfilled, 
    both two timers will be reset, and the procedure will be renewed to serve the next request.
    In Fig.~\ref{fig:policy}, $d_1$ will be overheard and loaded into the cache at epoch 8,
    if the second request for $d_1$ is not generated before epoch 8.
\end{enumerate} 
\begin{figure}[h]
    \vspace{-2mm}
    \centering
    \includegraphics[width=7.4cm]{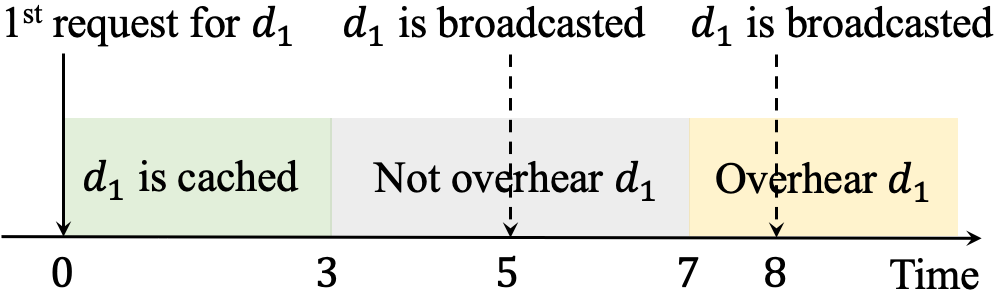}
    \vspace{-2mm}
    \caption{\color{black}TTL-based caching and overhearing policy.}
    \vspace{-2mm}
    \label{fig:policy}
\end{figure}

It is easy to observe that whether the next request for $d_i$ is a hit or a miss 
depends on when it arrives.
Specifically, we will use the example in Fig.~\ref{fig:policy} to illustrate this process.
\begin{enumerate}
    \item Request before eviction: 
    If the second request for $d_1$ arrives before the caching timer expires (i.e., epoch 3), 
    then it is a cache hit since $d_1$ has not been evicted yet. 
    \item Request during the deaf period: 
    If the second request for $d_1$ arrives in the period [3, 7], 
    then a cache miss occurs. 
    To serve the request, $d_1$ has to be fetched from the backend storage. 
    \item Request before overhearing: 
    If the second request for $d_1$ arrives in the period [7, 8],
    then $d_1$ is still a cache miss. 
    \item Request after overhearing:
    If the second request for $d_1$ arrives after epoch 8,
    then it is a cache hit and the request can be served from the cache. 
\end{enumerate}
Note that how to strategically choose $\tau_i$ and $\omega_i$ parameters 
is crucial to efficiently utilize the limited cache space.
For example, if $\omega_i$ is very large, 
the next request for $d_i$ is very likely to be generated
before the deaf period expires, and 
the hit ratio will be low. 
Instead, if $\omega_i$ takes a small value, we could
overhear and load it into the cache at the very early stage,
but it would be a waste of cache space if $d_i$ will not be requested in the near future.
}

To further expand the design space, 
we allow possible randomization for the item policies.
Let 
\begin{align}
    & \Pi^\trco  \eqd \nonumber \\
    & \hspace{0mm} \Big\{\pi^\trco \big(
        \big(q^{(1)},\cdots, q^{(n)}\big), 
        \big(\tau^{(1)},\cdots, \tau^{(n)}\big), 
        \big(\omega^{(1)},\cdots, \omega^{(n)}\big)
        \big):&  \nonumber \\
    & 0\leq q^{(j)}\leq 1, \sum_{j=1}^n q^{(j)} = 1, 
    0\leq \tau^{(j)} \leq \omega^{(j)}, \nonumber\\
    & 1\leq j\leq n, n\in \N \Big\} \nonumber
\end{align}
denote the set of all possible randomized item policies based on $\policyset^\tco$,
where the superscript ``rco'' stands for randomized caching and overhearing.
Each randomized item policy 
is a randomization of $n$ deterministic item policies in $\policyset^\tco$,
where $n$ could be any positive integer
and $q^{(j)}$ is the probability to apply the $j$-th deterministic item policy 
$\policy^\tco(\tau^{(j)},\omega^{(j)})$.
Suppose $d_i$ is served by 
$\policy^\trco(\mathbf{q}_i, \boldsymbol{\tau}_i, \boldsymbol{\omega}_i)$
with $\mathbf{q}_i=(q_i^{(1)},\cdots, q_i^{(n)})$,
$\boldsymbol{\tau}_i=(\tau_i^{(1)},\cdots, \tau_i^{(n)})$
and $\boldsymbol{\omega}_i=(\omega_i^{(1)},\cdots, \omega_i^{(n)})$.
Every time the data item $d_i$ is requested, 
a deterministic item policy $\pi^\tco(\tau^{(j)}_i, \omega^{(j)}_i)$
will be selected with a probability $q_i^{(j)}$ and applied to update the cache content, 
$1\leq j\leq n$. 
Notably, each item $d_i$ can be served by 
its customized item policy with carefully-selected parameters
$\mathbf{q}_i$, $\boldsymbol{\tau}_i$ and $\boldsymbol{\omega}_i$.
And the caching decisions for different data items are independent.
Thus, we can analyze the hit ratio of each data item separately. 
Since the caches are homogeneous, we assume that 
the item policies for the same data item are identical on difference caches.

\subsection{Hit Ratio Maximization}\label{sec:HRM}
For each data item $d_i$, $1\leq i \leq N$, 
we define its expected hit ratio 
achieved by an item policy $\policy_i$ on an edge cache as
\begin{align}
     & h_i(\policy_i) \nonumber \\
     & \eqd \expect \left[ \lim_{T \to \infty} 
     \frac{\text{Number of hits for $d_i$ during $[0, T]$ under $\policy_i$}}
     {\text{Number of requests for $d_i$ during $[0, T]$}} \right]. \nonumber
\end{align}
Let $p_i$ denote the probability that a request is for data item $d_i$, $1\leq i\leq N$,
which can be calculated as
\begin{align}\label{eq:p_i}
    p_i = \frac{1}{s_i + 1/\beta_i} \left/\sum_{j=1}^N\frac{1}{s_j + 1/\beta_j}\right. .
\end{align}
Since the demand dynamics are homogeneous across different caches,
the overall expected hit ratio of all $M$ caches is 
equal to the expected hit ratio of a single cache, 
which can be expressed by $\sum_{i=1}^N p_ih_i(\policy_i)$. 
We would like to maximize the overall expected hit ratio 
under the cache capacity constraint. 
For an edge cache, define the expected cache occupancy for $d_i$ as 
\begin{align}
    r_i(\policy_i) \eqd  \expect \Big[ 
        \lim_{T \to \infty} \frac{1}{T} \cdot 
        &\big(\text{Duration when $d_i$ is stored} \nonumber \\
        &\hspace{4mm}\text{ in the cache during $[0,T]$}\big)
    \Big], \nonumber
\end{align}
which characterizes the average cache space used for storing $d_i$.
The cache capacity constraint states that 
the expected cache occupancy of all data items should not 
exceed the cache size.
{\color{black}
Notably, the cache capacity constraint considers the average cache occupancy 
rather than the real-time cache occupancy.
We adopt the cache capacity constraint in an average sense for the following reason:
\begin{itemize}
    \item It simplifies the analysis of the hit ratio maximization problem, so we could derive efficient
        caching policies with provable performance.  
    \item The caching policy obtained under the average cache capacity constraint could be easily 
     generalized to satisfy to real-time cache capacity constraint with minor performance regressions. 
     For example, 
     if loading an overheard item into cache will violate the real-time cache capacity constraint, 
     we have the option to reject the operation. 
\end{itemize} 
}

Formally, we propose the hit ratio maximization problem 
\begin{align}\label{eq:opt_TTL_1}
    \max_{\policy_i} \hspace{12mm}& \sum_{i=1}^N p_i \cdot h_i(\policy_i)   \\
    \text{subject to} \hspace{8mm} & \policy_i \in \policyset^\text{rco}, \hspace{9mm} 1\leq i\leq N,\nonumber \\
    &\sum_{i=1}^N r_i(\policy_i) \leq b. \nonumber
\end{align}
The objective is to maximize the overall hit ratio of an edge cache by
selecting the optimal item policy for each data item from 
the item policy set $\policyset^\trco$.
Since the demand dynamics are homogeneous across different caches,
applying the optimal policy of the proposed problem to all $M$ caches 
should maximize the overall hit ratio of the entire system.  

{\color{black}
Note that the optimal caching policies will not changer over time, instead it 
captures the statistics of the demand dynamics and maximizes the expected overall hit ratio.
However, if the statistics used to characterize the demand dynamics (i.e., $s_i$ and $\beta_i$) are time varying, 
it would necessitate the optimal policy to change over time. 
Such a topic, however, falls outside the scope of this paper's discussion.
}
Next, we will investigate this problem under 
the two different overhearing settings, i.e., time-driven 
and event-driven overhearing.

\section{Edge Caching with Time-Driven Overhearing}
\label{sec:single}

In this section, we consider the time-driven overhearing scenario
where the overhearing processes of the users 
are governed by independent processes from the request dynamics. 
This is particularly relevant to scenarios 
where the base station broadcasts the items at a regular rate.  
In particular, we assume that the overhearing opportunity of the data item $d_i$ 
is a Poisson process with rate $\lambda_i$, $1\leq i \leq N$.
Once an overhearing opportunity is generated, 
each cache may decide to load the overheard item into its cache or not, depending on 
the caching and overhearing policy.

\subsection{Hit Ratios and Cache Occupancies}
Since the caches are homogeneous
and the overhearing process is independent of the number of caches, 
it suffices to analyze the system with a single cache. 
To simplify the notations, we use $h_i^\tco(\tau_i,\omega_i) \eqd h_i(\pi^\tco(\tau_i,\omega_i))$
and $r_i^\tco(\tau_i,\omega_i) \eqd r_i(\pi^\tco(\tau_i,\omega_i))$
to denote the expected hit ratio and cache occupancy of the data item $d_i$,
if it is served by the deterministic item policy $\pi^\tco(\tau_i,\omega_i)$.
In Theorem~\ref{theorem:hit_ratio}, we characterize $h_i^\tco(\tau_i,\omega_i)$ and 
$r_i^\tco(\tau_i,\omega_i)$ explicitly.

\begin{theorem}\label{theorem:hit_ratio}
    For time-driven overhearing, if 
    the data item $d_i$ is served by a deterministic item policy $\pi^\tco(\tau_i,\omega_i) \in \policyset^\tco$,
    we have

    \noindent
    (1) for $\tau_i \leq \omega_i \leq s_i$,
    \begin{align}
        & h_i^\tco(\tau_i, \omega_i) = 1- \frac{\beta_i}{\lambda_i+\beta_i} \exp(-\lambda_i(s_i-\omega_i)), \nonumber \\
        & r_i^\tco(\tau_i, \omega_i) = \frac{1}{\expect[X_i]} \Big(
            \frac{\beta_i}{\lambda_i(\lambda_i+\beta_i)} \exp(-\lambda_i(s_i-\omega_i))\nonumber \\
            & \hspace{27mm} + s_i - \omega_i - \frac{1}{\lambda_i} + \frac{1}{\beta_i}
        \Big), \nonumber
    \end{align}
        
    \noindent
    (2) for $\tau_i \leq s_i  \leq \omega_i$,
    \begin{align}
        & h_i^\tco(\tau_i, \omega_i) = \frac{\lambda_i}{\lambda_i+\beta_i} \exp(-\beta_i(\omega_i-s_i)), \nonumber \\
        & r_i^\tco(\tau_i, \omega_i) = \frac{1}{\expect[X_i]} \cdot \frac{\lambda_i}{\beta_i(\lambda_i+\beta_i)} 
          \exp(-\beta_i(\omega_i-s_i)), \nonumber
    \end{align}

    \noindent
    (3) for $s_i \leq \tau_i \leq \omega_i$,
    \begin{align}
        & h_i^\tco(\tau_i, \omega_i) = 1 - \exp(-\beta_i(\tau_i-s_i)) \nonumber \\
            & \hspace{33mm} + \frac{\lambda_i}{\lambda_i+\beta_i} \exp(-\beta_i(\omega_i-s_i)), \nonumber \\
        & r_i^\tco(\tau_i, \omega_i) = 1 + 
        \frac{1}{\expect[X_i]} \cdot \Big( 
            \frac{\lambda_i}{\beta_i(\lambda_i+\beta_i)} \exp(-\beta_i(\omega_i-s_i)) \nonumber \\
            & \hspace{33mm}- \frac{1}{\beta_i} \exp(-\beta_i(\tau_i-s_i))
        \Big), \nonumber
    \end{align}
    where $X_i$ is defined as the inter-request time for $d_i$ and $\expect[X_i] = s_i + 1/\beta_i$.
\end{theorem}

The proof of Theorem~\ref{theorem:hit_ratio} is presented in Appendix~\ref{sec:pf_hit_ratio}. 
Note that for a fixed $\tau_i<s_i$, $h_i^\tco$ is concave with respect to $\omega_i$
when $\omega_i<s_i-\tau_i$ and convex when $\omega_i\geq s_i-\tau_i$.
Therefore, the original problem (\ref{eq:opt_TTL_1}) is a nonconvex optimization problem,
which is difficult to solve in general.
However, for this specific problem, 
we could find the global optimum by
exploiting an informative structure of the optimal solution,
which will be presented in Section~\ref{sec:structure}.

Next, we will leverage Theorem~\ref{theorem:hit_ratio} to calculate the 
expected hit ratio and cache occupancy for randomized item policies. 
Consider a randomized item policy $\pi_i^\trco(\qq_i, \boldsymbol{\tau}_i, \boldsymbol{\omega}_i)$ 
for the data item $d_i$ with
$\qq_i = (q_i^{(1)}, \cdots, q_i^{(n)})$, 
$\boldsymbol{\tau}_i = (\tau_i^{(1)}, \cdots, \tau_i^{(n)})$,
$\boldsymbol{\omega}_i = (\omega_i^{(1)}, \cdots, \omega_i^{(n)})$.
Let $h_i^\trco$ and $r_i^\trco$ denote the expected hit ratio and the expected
cache occupancy achieved by $\pi_i^\trco$. 
We derive the explicit expression for $h_i^\trco$  and $r_i^\trco$ 
in the following theorem.

\begin{theorem}\label{theorem:hit_ratio_random}
    For time-driven overhearing, if 
    the data item $d_i$ is served by a randomized item policy 
    $\pi^\trco(\qq_i, \boldsymbol{\tau}_i, \boldsymbol{\omega}_i)$,
    then we have
    \begin{align}
        & h_i^\trco (\qq_i, \boldsymbol{\tau}_i, \boldsymbol{\omega}_i) 
        = \sum_{j=1}^n q_i^{(j)} \cdot 
            h_i^\tco(\tau_i^{(j)}, \omega_i^{(j)}), \nonumber \\
        & r_i^\trco (\qq_i, \boldsymbol{\tau}_i, \boldsymbol{\omega}_i)
        = \sum_{j=1}^n q_i^{(j)} \cdot
            r_i^\tco(\tau_i^{(j)}, \omega_i^{(j)}), \nonumber
    \end{align}
    where $h_i^\tco(\tau_i^{(j)}, \omega_i^{(j)})$, $r_i^\tco(\tau_i^{(j)}, \omega_i^{(j)})$
    can be explicitly characterized by Theorem~\ref{theorem:hit_ratio}.
\end{theorem}
It is shown that the expected hit ratio and cache occupancy
of a randomized item policy can be calculated as 
the linear combination of the ones of its basic policies. 

\subsection{Informative Structure of Optimal Policies}
\label{sec:structure}

In this section, we will prove a special structure 
of the optimal caching and overhearing policies, 
which significantly simplifies the optimization problem. 
Intuitively, for each data item, 
an item policy utilizes the cache space as the resource to 
achieve a high hit ratio which can be viewed as the revenue.
Thus, to evaluate how efficient an item policy is, 
a straightforward approach is to characterize the relationship
between the hit ratio and the cache occupancy achieved by it. 

For each data item $d_i$,
the hit ratio and the cache occupancy that can be achieved by 
some item policy in the set $\policyset^\trco$
can be described by an achievable region. 
Formally, define the achievable region for the data item $d_i$ as 
\begin{align} 
    &\mathcal{R}_i^\tco = \{(r, h): 
    \text{there exists $\policy_i \in \policyset^\trco$ such that } \nonumber \\
    &\hspace{2mm}\text{$\policy_i$ achieves a cache occupancy $r$ and a hit ratio $h$ 
    for $d_i$}\}. \nonumber
\end{align}

To better characterize the achievable region $\mathcal{R}_i^\tco$,
we investigate two specific item policies.
Define
\begin{align}
    \pi^\tc(\tau) \eqd \pi^\tco(\tau,\infty) 
    \hspace{4mm} \text{and} \hspace{4mm}
    \pi^\too(\omega) \eqd \pi^\tco(0, \omega). \nonumber
\end{align}
The \emph{caching-only} item policy $\pi^\tc(\tau)$ is a specific case of 
$\policy^\tco(\tau, \omega)$
with a caching TTL $\tau$ and an infinite overhearing TTL, 
i.e., never overhearing.
The \emph{overhearing-only} item policy $\pi^\too(\omega)$ is a specific case of 
$\policy^\tco(\tau, \omega)$
with an overhearing TTL $\omega$ and a caching TTL zero, 
i.e., evicting the item immediately after serving its request.

If we restrict the item policy to be selected from the item policy set
$\{\pi^\tc(\tau): \tau\geq 0\}$
or $\{\pi^\too(\omega):\omega\geq 0 \}$,
then the achievable region will degenerate to a curve, 
since it can be parameterized in one variable (i.e., $\tau$ or $\omega$, respectively). 
Therefore, the hit ratio achieved by $\pi^\tc(\tau)$ and $\pi^\too(\omega)$
can be viewed as a function of the cache occupancy. 
Formally, for each data item $d_i$,
define
\begin{align*}
    h_i^\tc(r) \eqd h_i^\tco(\tau_i, +\infty),
\end{align*}
where $\tau_i$ is selected such that $r^\tco_i(\tau_i, +\infty) = r$, and
\begin{align*}
    & h_i^\too(r) \eqd h_i^\tco(0, \omega_i)
\end{align*}
where $\omega_i$ is selected such that $r_i^\tco(0, \omega_i) = r$. 
For an item $d_i$,
$h_i^\tc(r)$ (respectively, $h_i^\too(r)$) is the expected hit ratio achieved by 
the item policy $\pi^\tc(\tau_i)$ (respectively, $\pi^\too(\omega_i)$) such that 
the average cache space used to store $d_i$ is $r$.
The $h_i^\tc(r)$ and $h_i^\too(r)$ functions 
can be easily derived based on Theorem~\ref{theorem:hit_ratio}.
We plot these two functions in Fig.~\ref{fig:h_r}.

Note that, by setting $\tau_i=+\infty$, the item policy $\policy^\tc(\tau_i)$
can achieve the maximal hit ratio $1$ and the maximal cache occupancy $1$.
Under this setting, the data item will always be stored in the cache. 
By setting $\omega_i=0$, the item policy $\policy^\too(\omega_i)$
can achieve the maximal hit ratio $h_i^\too(r_i^\tco(0,0)) = h_i^\tco(0,0)$, 
which is smaller than $1$.
The reason is that the request of $d_i$ may arrive before the overhearing opportunities
even when we set $\omega_i = 0$.
Based on Theorem~\ref{theorem:hit_ratio_random},
any points on the line segment connecting 
$(1,1)$ and $(r_i^\tco(0,0), h_i^\tco(0,0))$ 
can be achieved by a randomization of $\policy^\tc(+\infty)$ and $\policy^\too(0)$.

\begin{figure}[t]
    \centering
    \vspace{0mm}
    \includegraphics[width=7cm]{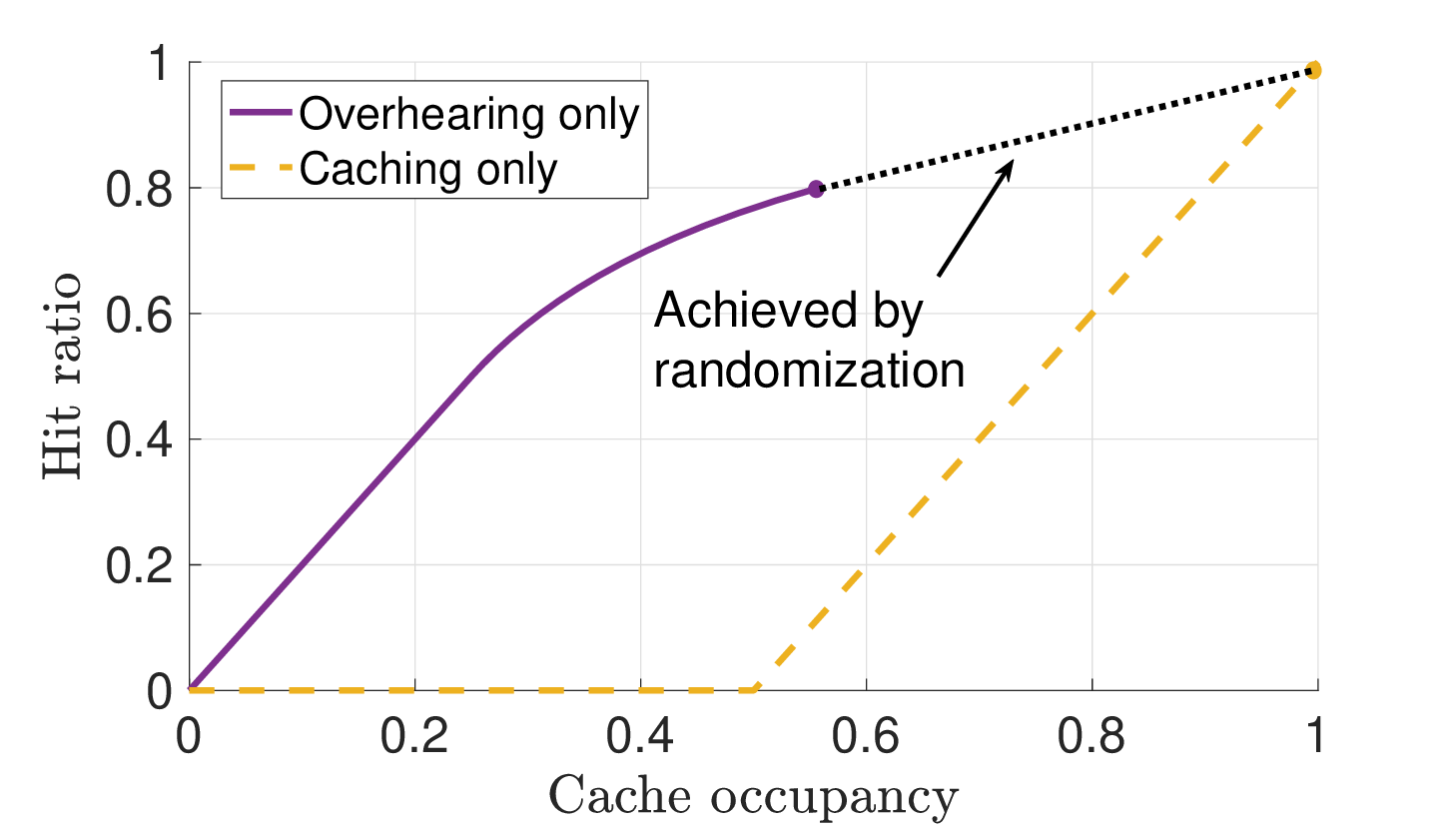}
    \vspace{-0mm}
    \caption{Hit ratio and cache occupancy achieved by the caching-only item policies and 
    the overhearing-only item policies 
    with $s_i = \beta_i = \lambda_i = 1$.}
    \vspace{-2mm}
    \label{fig:h_r}
\end{figure}

Formally, we define a randomized caching and overhearing item policy set 
\begin{align} \label{eq:policy_set_1}
    & \widetilde{\policyset}^\trco = \{\policy^\too(\omega): \omega \geq 0\} \nonumber \\
    & \cup
    \{\policy^\trco((q,1-q),(+\infty,0),(+\infty,0)): 0 \leq q \leq 1\}.
\end{align}
The item policy set $\widetilde{\policyset}^\trco$ contains all overhearing-only 
item policies and all possible randomizations of $\policy^\tc(+\infty)$ and $\policy^\too(0)$.
The achievable region of $\widetilde{\policyset}^\trco$ can also be characterized by a curve.
We define the $h_i^\text{rco}(r)$ function as 
the hit ratio achieved by an item policy from $\widetilde{\policyset}^\trco$
when the cache occupancy is $r$.
For $r\in[0,r^\tco(0,0)]$, we have $h_i^\text{rco}(r)= h_i^\text{o}(r)$.
For $r\in(r^\tco(0,0),1]$, $h_i^\text{rco}(r)$ is the line segment 
connecting the points $(r_i^\tco(0,0), h_i^\tco(0,0))$ 
and $(1, 1)$.
In Fig.~\ref{fig:h_r}, $h_i^\text{rco}(r)$ is the curve labeled by ``overhearing only''
and the line segment achieved by randomization. 
Notably, every point on the curve $h_i^\text{rco}(r)$ corresponds to 
exact one policy in the set $\widetilde{\policyset}^\trco$, and vice versa. 
Based on Theorems~\ref{theorem:hit_ratio} and \ref{theorem:hit_ratio_random},
we can easily calculate the parameters (i.e., $\omega$, $q$) for 
an item policy from $\widetilde{\policyset}^\trco$ that achieves a given 
hit ratio or cache occupancy. 

Next, we show an insightful characteristic of 
the achievable region $\mathcal{R}_i^\tco$.
\begin{lemma}\label{lemma:boundary}
    For any $(r, h)\in \mathcal{R}_i^\tco$, we have
    $h \leq h_i^\trco(r)$. 
\end{lemma}
The proof is presented in Appendix~\ref{sec:proof_lemma_boundary}.
Lemma~\ref{lemma:boundary} shows that
the upper boundary of the achievable region $\mathcal{R}_i^\tco$
is characterized by the function $h_i^\trco(r)$,
based on which, 
we can prove an informative structure of the optimal policies.
\begin{theorem}\label{theorem:opt_structure}
    For time-driven overhearing, there must exist a caching and overhearing policy 
    $\policyvec^* = (\policy_1^*, \cdots, \policy_N^*)$ which is the optimal solution 
    of problem~(\ref{eq:opt_TTL_1}) and satisfies $\policy_i^* \in \widetilde{\policyset}^\trco$ 
    for any $1\leq i \leq N$.
\end{theorem}
Theorem~\ref{theorem:opt_structure} is a direct application of Lemma~\ref{lemma:boundary}.
It shows that an optimal solution of 
problem~(\ref{eq:opt_TTL_1}) can be found from the 
set $\widetilde{\policyset}^\trco$,
which significantly narrows the design space. 
By leveraging this informative structure,
we solve the optimal caching and overhearing policy
in the next section.

\subsection{Optimal Policy for Time-Driven Overhearing}
\label{sec:opt_time_driven}
Directly replacing the policy set in problem~\ref{eq:opt_TTL_1} with $\widetilde{\policyset}^\trco$ 
will still result in a nonconvex optimization. 
Instead, we will solve this problem by following two steps. 

\noindent 
\textbf{Step 1}: Solve the optimal $r_i^*$'s of the following problem (\ref{eq:opt_TTL_2})
\begin{align}
    \max_{r_i} \hspace{12mm}& \sum_{i=1}^N  p_i \cdot h_i^\text{rco}(r_i)   
        \hspace{-20mm} & \label{eq:opt_TTL_2} \\
    \text{subject to} \hspace{8mm} 
    & 0 \leq r_i \leq 1, &  \hspace{-10mm} 1\leq i\leq N, \nonumber \\
    &\sum_{i=1}^N r_i \leq b. & \hspace{-0mm} \nonumber
\end{align}
Note that the original problem~(\ref{eq:opt_TTL_1}) is trying to find the optimal policy parameters 
(i.e., $\qq$, $\boldsymbol{\tau}$ and $\boldsymbol{\omega}$).
In Step~1, the original optimization 
problem in the domain of policy parameters is converted into 
a new problem in the domain of the cache occupancies (i.e., $\boldsymbol{r}$).
Recall that $h_i^\trco(r_i)$ captures the relationship between
the hit ratio and the cache occupancy for item policies in the set $\widetilde{\policyset}^\trco$.
Using Theorem~\ref{theorem:opt_structure}, 
we can prove that the optimal occupancies found by problem~(\ref{eq:opt_TTL_2}) are actually 
the occupancies achieved by the optimal item policies of the original problem~(\ref{eq:opt_TTL_1}).
More importantly, 
Theorems~\ref{theorem:hit_ratio} and~\ref{theorem:hit_ratio_random} indicate that 
the $h_i^\trco(r)$ functions are convex, $1\leq i\leq N$.
Therefore, the optimization problem~(\ref{eq:opt_TTL_2}) is convex and  
can be solved using standard tools
(e.g., KKT conditions and the water filling algorithm~\cite{boyd2004convex}).

\noindent 
\textbf{Step 2}: Once the optimal solution $r^*_i$'s of Step 1 is solved, 
then based on Theorems~\ref{theorem:hit_ratio} and 
\ref{theorem:hit_ratio_random},
we can easily find the item policies from the set $\widetilde{\policyset}^\trco$
that achieve $r^*_i$'s.
And these item policies form an optimal solution of the original hit ratio maximization
problem~(\ref{eq:opt_TTL_1}).
We use $\omega^*_i$'s and $q^*_i$'s to denote the parameters for these 
item policies.
An optimal policy for time-driven overhearing is formally proposed as follows. 

\noindent
\textbf{Caching and overhearing policy for time-driven overhearing
($\policyvec^\tcot$):}
Serve each data item $d_i$, $1\leq i\leq N$,  by a randomized item policy 
$\policy^\trco(\qq_i, \boldsymbol{\tau}_i, \boldsymbol{\omega}_i)$
where $\qq_i=(q_i^*,1-q_i^*)$, $\boldsymbol{\tau}_i=(+\infty,0)$
and $\boldsymbol{\omega}_i=(+\infty,\omega_i^*)$.



The proposed optimal policy reveals the following insights:
\begin{enumerate}
    \item We should
    either evict a data item immediately after serving a request for it
    (i.e., set $\tau_i = 0$) 
    or always store it in the cache (i.e, set $\tau_i = +\infty$).
    Setting $\tau_i\in (0,+\infty)$ will be suboptimal. 
    \item If we decide to bring an item back into the cache 
    by overhearing (i.e., set $\omega_i < +\infty$), 
    then that item should be evicted immediately after serving each request for it
    (i.e., set $\tau_i = 0$).
\end{enumerate}
These insights will guide us 
in the more complex scenario of event-driven overhearing, 
which is tackled next.

\section{Edge Caching with Event-Driven Overhearing}
\label{sec:multi}

In this section, we consider a more realistic overhearing setting.
When a cache misses a data item, the base station will send
the data item to it over broadcast channels. 
Meanwhile, the other caches could overhear and decide 
whether they would like to store this data item.  
Since the overhearing opportunities are triggered by cache misses, 
the event-driven overhearing process depends on the caching decisions
as well as the number of caches in the system.
It can be easily verified that the overhearing process are not Poisson 
under this setting. 
As a result, the analysis for time-driven overhearing cannot be directly 
applied for the event-driven scenario.

\subsection{Hit Ratios and Cache Occupancies}
\label{sec:hit_ratio_event_driven}
It is difficult to derive the hit ratio and the 
cache occupancy for a general caching and overhearing policy,
since the overhearing process is not tractable under this setting. 
However, we are able to characterize a few key properties 
for some specific policies, which could inspire us to 
design a provably good policy. 

Similar to the notations in Section~\ref{sec:single}, 
we still use $h_i^\too(\cdot)$, $h_i^\tc(\cdot)$, $h_i^\tco(\cdot)$, $h_i^\trco(\cdot)$
to denote the hit ratios achieved by the item policies 
$\policy^\too$, $\policy^\tc$, $\policy^\tco$, $\policy^\trco$
for the data item $d_i$, respectively.
The same rules will also be applied to the notations for cache occupancies.
However, the expression of these functions will be different from those 
in Section~\ref{sec:single}, since the overhearing processes have been changed.

\begin{lemma}\label{lemma:property2}
    Consider the event-driven overhearing. 
    If $d_i$ is served by the item policy $\pi^\too(\omega_i)$,
    then we have $$h_i^\too(r) = (\beta_is_i+1)  r $$
    for $0\leq r \leq r_i^\tco(0,s_i)$.
    If $d_i$ is served by the item policy $\pi^\tc(\tau_i)$,
    then its hit ratio and occupancy are exactly the same as those 
    achieved by $\pi^\tc(\tau_i)$ for time-driven overhearing, 
    and can be directly calculated using 
    Theorem~\ref{theorem:hit_ratio}.
\end{lemma}

In Lemma~\ref{lemma:property2},
we analyze the item policies $\pi^\tc(\tau_i)$
and $\pi^\too(\omega_i)$ under event-driven overhearing.
The proof is presented in Appendix~\ref{sec:proof_lemma_property2}.
For the caching-only item policy $\pi^\tc(\tau_i)$,
the hit ratio and the cache occupancy 
are exactly the same as the ones under time-driven overhearing,
since $\pi^\tc(\tau_i)$ 
sets $\omega_i = +\infty$ and is independent of the overhearing process.
For the overhearing-only item policy $\pi^\too(\omega_i)$,
we show that the hit ratio of $d_i$ is a linear function 
with respect to the cache occupancy when $0\leq r \leq r^\tco(0,s_i)$,
or, equivalently when $\omega_i \geq s_i$.
When $\omega_i < s_i$, the overhearing-only item policy becomes intractable.
\begin{figure}[t]
    \centering
    \includegraphics[width=7cm]{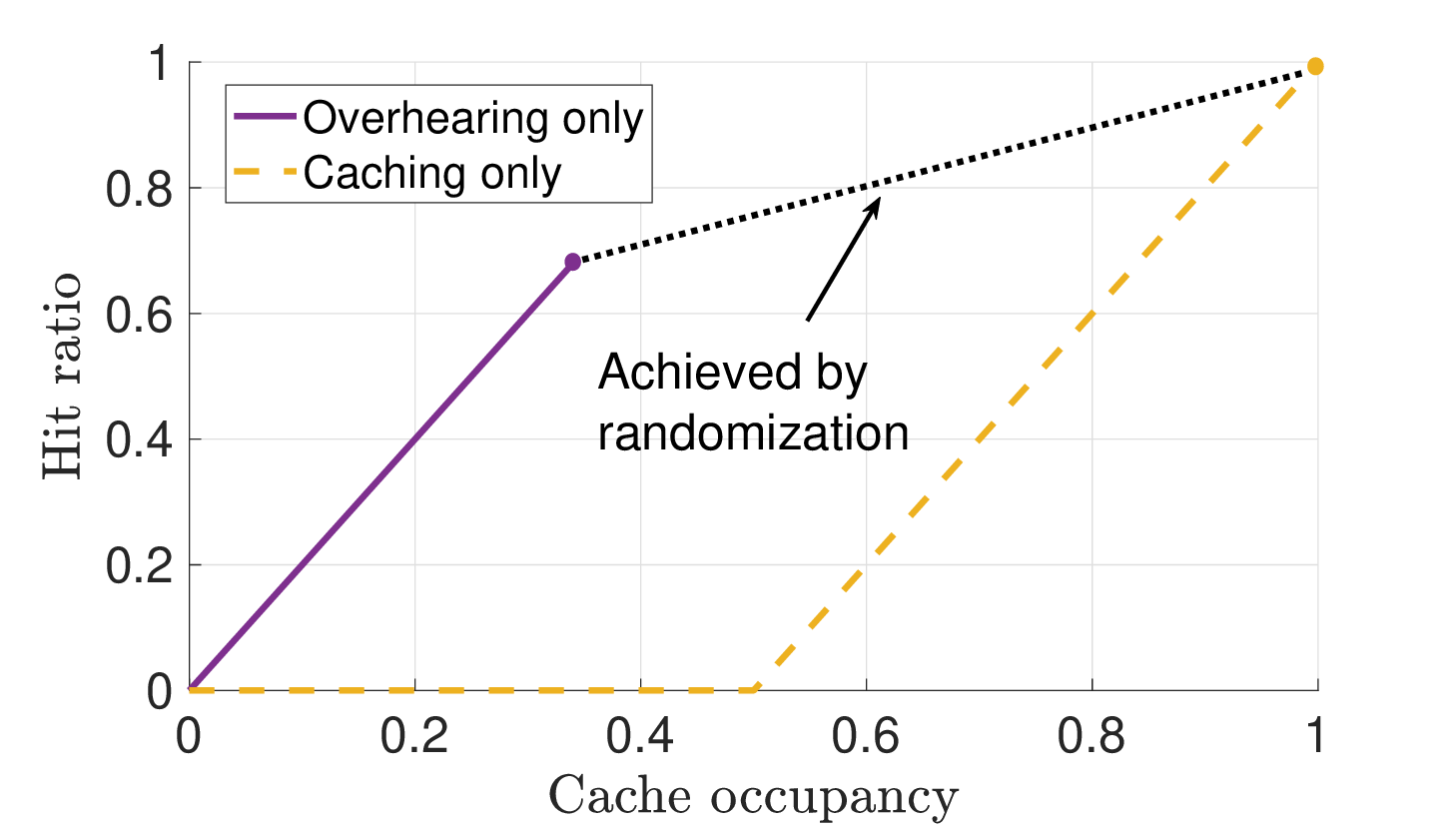}
    \vspace{-0mm}
    \caption{Hit ratio and cache occupancy achieved by the caching-only item policies and 
    the overhearing-only item policies
    with $s_i = \beta_i = 1$ and $M = 10$.}
    \vspace{-0mm}
    \label{fig:h_r_event_driven}
\end{figure}
We plot the hit ratio and the cache occupancy that can be achieved by 
$\policy^\too(\omega_i)$ with $\omega_i \geq s_i$ and 
$\policy^\tc(\tau_i)$ with $\tau_i \geq 0$ 
in Fig.~\ref{fig:h_r_event_driven}.

In Lemma~\ref{lemma:property2} we characterize the relationship between hit ratios and cache occupancies
for event-driven overhearing,
but we are not able to analytically solve the parameter $\omega_i$
that achieves a given cache occupancy $r$.
To address this issue, we first assume that
the policy parameter $\omega_i$ to achieve any $r\leq r^\tco(0,s_i)$
is solvable. 
With this assumption, we will propose provably good policies 
in Section~\ref{sec:opt_general}. 
Then, in Section~\ref{sec:discuss}, 
we will discuss how to implement these policies in real practice 
without the proposed assumption.

\subsection{Asymptotically Optimal Policy for Event-Driven Overhearing}
\label{sec:opt_general}
Although the hit ratio and the cache occupancy under event-driven overhearing are
not fully tractable, we could still design provably good polices by leveraging 
the insights obtained from the optimal structure under time-driven overhearing. 
The general idea is to first construct a shrunken policy set, which contains 
less item policies but retains some tractability under event-driven overhearing. 
Then, we will find the best policy from the shrunken policy set and 
analytically characterize its performance.

For each data item $d_i$, $1\leq i\leq N$,
define a policy set 
\begin{align} \label{eq:policy_set_2}
    & \widehat{\policyset}_i^\trco = 
    \{\policy^\too(\omega_i): \omega_i \geq s_i\} \nonumber \\
    & \cup
    \{\policy^\trco((q_i,1-q_i),(+\infty,0),(+\infty,s_i)): 0 \leq q_i \leq 1\}.
\end{align}
The set $\widehat{\policyset}_i^\trco$ contains overhearing-only item policies
$\policy^\too(\omega_i)$ with $\omega_i\geq s_i$
and all possible randomizations of $\policy^\too(s_i)$ and $\policy^\tc(+\infty)$.
The reason to construct such policy sets is that we could characterize the relationship
between hit ratios and cache occupancies for these item policies based on Lemma~\ref{lemma:property2}.
Instead of solving the original problem~(\ref{eq:opt_TTL_1}), 
we would like to find the best policy from the shrunken policy sets by solving the following problem
\begin{align}\label{eq:opt_TTL_6}
    \max_{\policy_i} \hspace{12mm}& \sum_{i=1}^N p_i \cdot h_i(\policy_i)   \\
    \text{subject to} \hspace{8mm} & \policy_i \in \widehat{\policyset}_i^\trco, \hspace{9mm} 1\leq i\leq N,\nonumber \\
    &\sum_{i=1}^N r_i(\policy_i) \leq b. \nonumber
\end{align}

Define $\widehat{h}_i^\trco(r)$ as the hit ratio of $d_i$ 
achieved by an item policy from the set 
$\widehat{\policyset}_i^\trco$
such that the cache occupancy is $r$.
We have $\widehat{h}_i^\trco(r) = (s_i\beta_i+1)r$ for $0\leq r \leq r_i^\tco(0,s_i)$,
and 
\begin{align*}
    \widehat{h}_i^\trco(r) = \frac{1-(s_i\beta_i+1)r_i^\tco(0,s_i)} {1-r_i^\tco(0,s_i)} \cdot r 
    + \frac{s_i\beta_i r_i^\tco(0,s_i)}{1-r_i^\tco(0,s_i)}
\end{align*}
for $r_i^\tco(0,s_i)\leq r \leq 1$.
The $\widehat{h}_i^\trco(r)$ curve is illustrated in Fig.~\ref{fig:h_r_event_driven}
as the line segments achieved by overhearing only and randomization.
Every point on the curve $\widehat{h}_i^\trco(r)$
corresponds an item policy in the set $\widehat{\policyset}_i^\trco$,
and vice versa. 
To find the best item policies from $\widehat{\policyset}_i^\trco$,
we formulate the following optimization problem. 
\begin{align}\label{eq:opt_TTL_4}
    \max_{r_i} \hspace{12mm}& \sum_{i=1}^N  p_i \cdot \widehat{h}_i^\text{rco}(r_i)   
     \hspace{-20mm} &  \\
    \text{subject to} \hspace{8mm} 
    & 0 \leq r_i \leq 1, &  \hspace{-10mm} 1\leq i\leq N, \nonumber \\
    &\sum_{i=1}^N r_i \leq b. & \hspace{-0mm} \nonumber
\end{align}

Since $\widehat{h}_i^\trco(r)$ functions are convex,  
problem~(\ref{eq:opt_TTL_4}) is convex
and can be solved by the same approach that solves problem~(\ref{eq:opt_TTL_2}).
Let $r_i^*$'s denote the optimal solution to problem~(\ref{eq:opt_TTL_4}).
We can easily identify the item policy from the set $\widehat{\policyset}_i^\trco$
that achieves $r_i^*$, $1\leq i\leq N$. 
We propose a caching and overhearing policy as follows.

\noindent 
\textbf{Caching and overhearing policy for event-driven overhearing ($\policyvec^\tcoe$):} 
Let each data item $d_i$, $1\leq i\leq  N$, be served by the item policy from the set $\widehat{\policyset}_i^\trco$
that achieves the cache occupancy $r_i^*$, i.e., the solution of~(\ref{eq:opt_TTL_4}).

To analytically characterize the performance of the proposed policy $\policyvec^\tcoe$, 
we first introduce an upper bound for the achievable hit ratio. 
For a system consisting of $M$ caches,
let $h^*(M)$ denote the overall hit ratio achieved by 
the optimal solution of problem~(\ref{eq:opt_TTL_1}) under event-driven overhearing.
In the following lemma, we prove that $h^*(M)$ is upper bounded by a constant which is 
defined as $h^\text{upper}$.
\begin{lemma}\label{lemma:upper}
    Consider a system of $M$ caches where each cache has a size $b$.
    We have 
    \begin{align}\label{eq:upper_hit}
        & h^*(M) \nonumber \\
        & \leq \sum_{i=1}^K p_i 
            + p_{K+1} \left(\beta_{K+1} s_{K+1} + 1\right)
            \Big(b - \sum_{i=1}^{K} \frac{1}{\beta_is_i +1}\Big) \nonumber\\
        & \eqd h^\text{upper},
    \end{align}
    where $K$ is the integer such that 
    \begin{align}
        \sum_{i=1}^{K} \frac{1}{\beta_is_i +1} \leq b 
        < \sum_{i=1}^{K+1} \frac{1}{\beta_is_i +1} \label{eq:upper_K}. 
    \end{align}
\end{lemma}

The proof of Lemma~\ref{lemma:upper} is presented in Appendix~\ref{sec:proof_lemma_upper}.
The upper bound proposed in this lemma is actually the hit ratio 
achieved by an idealized policy. 
The idealized policy assumes that we could always overhear any data item 
at any time and attempts to find the best overhearing time based on the 
anticipated arrival time for the next request.
Along this direction, we could easily prove this lemma. 
And the detailed proof is omitted due to the page limit.

Let $h^\tcoe(M)$ denote the expected overall hit ratio achieved 
by $\policyvec^\tcoe$ in a system consisting of $M$ caches. 
We characterize the distance between $h^\tcoe(M)$ and $h^\text{upper}$ in the following 
theorem.

\begin{theorem}\label{theorem:asymp_opt}
    For the proposed policy $\policyvec^\tcoe$
    and $K$ defined in Equation~(\ref{eq:upper_K}), 
    we have, as the number of caches $M \to +\infty$,
    \begin{align}
        0 & \leq h^*(M) - h^\tcoe(M)  \nonumber \\
        & \leq h^\text{upper} - h^\tcoe(M)
        \leq
        \max_{1\leq i\leq K+1} 2\sqrt{\frac{\beta_is_i + 1}{M}}, \nonumber
    \end{align} 
    which implies that  
    \begin{align}
        \lim_{M\to +\infty} h^\tcoe(M) 
        = \lim_{M\to +\infty} h^*(M) 
        = h^\text{upper}.\nonumber
    \end{align}
\end{theorem}

In Appendix~\ref{sec:proof_theorem_asymp_opt},
we prove Theorem~\ref{theorem:asymp_opt} by showing 
that the expected inter-overhearing time for each date item
will converge to zero as $M$ goes to infinity. 
Theorem~\ref{theorem:asymp_opt} tells us that 
the proposed policy $\policyvec^\tcoe$ for event-driven overhearing setting 
is asymptotically optimal as the number of caches goes to infinity.
Intuitively, as the number of caches in the system increases, 
it will be more likely to overhear a data item.
The proposed policy could efficiently utilize the overhearing opportunities
and achieve asymptotically optimal performance. 

\subsection{Discussion on Implementation}
\label{sec:discuss}
Since the overhearing process is difficult to analyze, 
in order to design provably good policies, 
we previously assumed in Section~\ref{sec:hit_ratio_event_driven} that 
the cache occupancy achieved by
the item policy $\policy^\too(\omega_i)$, $\omega_i \geq s_i$,
can be analytically solved. 
Based on this tractability assumption, we propose and analyze 
$\policyvec^\tcoe$  in Section~\ref{sec:opt_general}.
In this section, we will discuss how to implement the proposed policies
without this assumption.

First, we note that it is impractical to estimate 
the cache occupancy for all possible $\policy^\too(\omega_i)$'s,
since $\omega_i$ can take any real numbers. 
In contrast, we will show that a good performance can be guaranteed 
by leveraging an accurate estimation of the cache occupancy achieved by 
a specific item policy $\policy^\too(s_i)$.


For the convex problem~(\ref{eq:opt_TTL_4}), KKT conditions show that 
the optimal solution satisfies that 
\begin{align}\label{eq:KKT_2}
    \sum_{i=1}^N \ind(r_i^* \neq r^\too_i(s_i) \text{ and } 0 < r_i^* < 1) \leq 1,
\end{align}
where $r^\too_i(s_i)$ is the cache occupancy achieved by the overhearing-only item policy
$\policy^\too(\omega_i)$ with $\omega_i=s_i$.
It indicates that there is at most one $r_i^*$ that takes a value other than 
$r^\too_i(s_i)$, 0 and 1.
In other word, except for one item policy, the other item policies in the optimal solution
must be the overhearing-only item policy $\policy^\too(s_i)$ or a randomization of $\policy^\too(s_i)$
and $\policy^\tc(+\infty)$.
As a result, we could further narrow down the policy set by considering 
$\policy^\too(s_i)$ and randomizations of $\policy^\too(s_i)$
and $\policy^\tc(+\infty)$, i.e., 
\begin{align}\label{eq:policy_set_3}
    \{\policy^\trco((q_i,1-q_i),(+\infty,0),(+\infty,s_i)): 0 \leq q_i \leq 1\},
\end{align}
Once we have a good estimation of $r^\too_i(s_i)$, 
all item policies in this set are tractable.

Therefore, to implement the caching and overhearing policy $\policyvec^\tcoe$ proposed in 
Section~\ref{sec:opt_general},
we could solve the hit ratio maximization problem based on the policy set (\ref{eq:policy_set_3})
rather than the one defined in (\ref{eq:policy_set_2}).
The solved policy is an approximation of $\policyvec^\tcoe$.
By applying (\ref{eq:KKT_2}), 
we can prove that the overall hit ratio achieved by this approximated policy
is within $1-1/b$ fraction of $\policyvec^\tcoe$, where $b$ is the cache size.


The remaining problem is how to estimate $r_i^\too(s_i)$
values. 
A simple solution would be to first run an estimation phase to 
approximate $r_i^\too(s_i)$ and then solve the 
modified policies using the estimated values. 
In the estimation phase, for each data item $d_i$, 
we may run $\pi^\too(s_i)$ for $T$ units of time,
and estimate $r_i^\too(s_i)$ by 
\begin{align}
    \bar{r}_i^\too(s_i) = 
    \left(\text{Duration of time when $d_i$ is cached}\right) / T.
    \nonumber
\end{align}

The proposed implementation solution could introduce performance losses
compared to the original policy $\policyvec^\tcoe$
due to the following two reasons:
\begin{enumerate}
    \item The implemented policy is an approximation of $\policyvec^\tcoe$ by considering 
    the policy set (\ref{eq:policy_set_3}) rather than (\ref{eq:policy_set_2}).
    \item The caching and overhearing decisions are biased since the estimation 
    $\bar{r}_i^\too(s_i)$ is not accurate. 
\end{enumerate}
However, these performance losses could be ignored as long as the cache size $b$
and the length of the estimation phase $T$ are sufficiently large.

\section{Generalization for Heterogeneous Demand Dynamics}
\label{sec:gen}

In the main paper, we focus on the homogeneous demand dynamics, 
where different users have the same demand pattern (i.e., $s_i$, $\beta_i$) for a given data item $d_i$. 
The obtained insights and theorems can be easily generalized for heterogeneous demands with minor modifications. 

We use $m$, $1\leq m \leq M$, to index an edge cache or 
the user served by the corresponding edge cache interchangeably, 
since each edge cache is assumed to serve a single user. 
Our first step is to extend the proposed ON-OFF processes 
to allow for different demand patterns of the same data item 
among different users. 
For the user $m$, we use the proposed ON-OFF process with
a OFF-period length $s_i^\um$ and ON-period request rate $\beta_i^\um$ 
to describe the demand dynamics of the data item $d_i$.
The popularity of $d_i$ for user $m$ can be evaluated by 
\begin{align*}
    p_i^\um = \frac{1}{s_i^\um + 1/\beta_i^\um} \left/\sum_{j=1}^N\frac{1}{s_j^\um + 1/\beta_j^\um}\right. .
\end{align*}
Define $\nu^\um$ as 
\begin{align}
    \nu^\um = \sum_{i=1}^N\frac{1}{s_i^\um + 1/\beta_i^\um} 
    \left/\sum_{m=1}^M\sum_{i=1}^N\frac{1}{s_i^\um + 1/\beta_i^\um}\right. .\nonumber
\end{align}
$\nu^\um$ represents the ratio of requests that are from user $m$.
We have $\sum_{m=1}^M \nu^\um = 1$.

Similar to the hit ratio and the cache occupancy defined in Section~\ref{sec:HRM},
we use $h_i^\um(\pi)$ and $r_i^\um(\pi)$ to denote the hit ratio and the cache occupancy 
of the data item $d_i$ in the edge cache $m$ achieved by the item policy $\pi$, respectively. 

Our goal is to find the optimal policy such that the overall hit ratio of the entire system 
is maximized. 
We formally propose the problem as follows.
\begin{align}\label{eq:opt_gen_1}
    \max_{\policy_i^\um, 1\leq m\leq M, 1\leq i\leq N} 
        \hspace{2mm}& \sum_{m=1}^M\sum_{i=1}^N \nu^\um p_i^\um \cdot h_i^\um(\policy_i^\um)  \\
    \text{subject to} \hspace{8mm} & \policy^\um_i \in \policyset^\text{rco}, 
         1\leq m\leq M, 1\leq i\leq N,\nonumber \\
    & \sum_{i=1}^N r_i^\um(\policy_i^\um) \leq b, \hspace{2mm} 1\leq m\leq M. \nonumber
\end{align}

\subsection{Time-Driven Overhearing}
The time-driven overhearing process is independent of the edge caching policy.
Under time-driven overhearing, the cache hit ratio and occupancy of an edge cache
are determined by its own policy and are independent of other caches.
Therefore, maximizing the overall hit ratio of the entire system is 
equivalent to maximizing the hit ratio of each edge cache separately. 
Formally, we can propose $M$ sub-problems, where the $m$-th problem is defined 
as follows.
\begin{align}\label{eq:opt_gen_2}
    \max_{\policy_i^\um, 1\leq i\leq N} \hspace{6mm}& \nu^\um p_i^\um \cdot \sum_{i=1}^N h_i^\um(\policy_i^\um)  \\
    \text{subject to} \hspace{8mm} & \policy^\um_i \in \policyset^\text{rco}, 
        \hspace{2mm} 1\leq i\leq N,\nonumber \\
    & \sum_{i=1}^N r_i^\um(\policy_i^\um) \leq b. \nonumber
\end{align}
Let $\{\policy_i^{\um*}: 1\leq i\leq N\}$ denote the optimal solution of the sub-problem~(\ref{eq:opt_gen_2}),
then $\{\policy_i^{\um*}: 1\leq i\leq N, 1\leq m\leq M\}$ would be the optimal solution of the original problem~(\ref{eq:opt_gen_1}).
Notably, solving the sub-problem~(\ref{eq:opt_gen_2}) is equivalent to solving problem~(\ref{eq:opt_TTL_1})
that was proposed for the homogeneous setting previously.
Therefore, our analysis for the homogeneous demand setting in Section~\ref{sec:single} is still valid for the 
heterogeneous setting with time-driven overhearing. 
And the proposed $\policyvec^T$ policy would be the solution to the sub-problem~(\ref{eq:opt_gen_2}),
where the policy inputs $s_i$, $\beta_i$ are replaced by $s_i^\um$ and $\beta_i^\um$, $1\leq i \leq N$.

\subsection{Event-Driven Overhearing}
\label{sec:genearlized_event_driven}
The event-driven overhearing process becomes more complicated when the demand dynamics are heterogeneous
across different edge caches. 
We are not able to split the original overall hit ratio maximization problem
to independent sub-problems designed for each cache, since the policy for one cache will impact the 
overhearing processes as well as the optimal decision of other caches. 

Due to the increased complexity, 
the policy $\policyvec^E$ proposed for homogeneous demand dynamics cannot 
be directly extended for heterogeneous settings. 
Fortunately, the key properties of hit ratios and cache occupancies
characterized in Lemma~\ref{lemma:property2} are still valid for each edge cache. 
In particular, for the edge cache $m$, Lemma~\ref{lemma:property2} holds if we replace $s_i$, $\beta_i$ by 
$s_i^\um$, $\beta_i^\um$.
Similar to the idea of proposing policy $\policyvec^E$ for homogeneous demands in Section~\ref{sec:multi},
we will leverage Lemma~\ref{lemma:property2} and the informative structure characterized in Section~\ref{sec:structure}
to design a provably good policy.

For the cache $m$ and the data item $d_i$, let $ind(i,m)$ be a reordering of the data index,
such that $ind(i,m)$ takes distinct integer values in $[1, N]$ for different input $i$,
and $\beta_{ind(i,m)}^\um \geq \beta_{ind(j,m)}^\um$ for any $1\leq i < j\leq N$.
Define the set of items $\D^\um = \{d_i: 1\leq ind(i,m) \leq K^\um\}$,
where $K^\um$ is an integer such that 
\begin{align} \label{eq:K_m}
    &\sum_{i=1}^{K^\um} \frac{1}{\beta_{ind(i,m)}^\um s_i^\um +1} \leq b, \\
    &\sum_{i=1}^{K^\um+1} \frac{1}{\beta_{ind(i,m)}^\um s_i^\um +1} > b.\nonumber
\end{align}
We say the data item $d_i$ is a popular data item for user $m$, 
if $ind(i,m)\leq K^\um$.
Define 
$$\mathcal{C}_i = \{m: ind(i,m) \leq K^\um\}$$
to be the set of users, for which $d_i$ is a popular data item. 
We propose the policy $\policyvec^\toe$ as follows.

\noindent
\textbf{Overhearing only policy for event-driven overhearing ($\policyvec^\toe$):} 
For each edge cache $m$, apply the overhearing only item policy $\policy^\too(s_i^\um)$
to serve the data item $d_i$, if $d_i \in \D^\um$.
Do not overhear or cache the data items that are not in $\D^\um$.
  
$\policyvec^\toe$ imitates the overhearing decisions of $\policyvec^\tcoe$,
which is proposed for homogeneous demand dynamics in Section~\ref{sec:opt_general}.
However, to simplify the analysis,
$\policyvec^\toe$ forgoes the caching only options of $\policyvec^\tcoe$.
For $M$ caches under event-driven overhearing, 
let $h^*(M)$ and  $h^\toe(M)$ denote the overall expected hit ratio achieved 
by the optimal policy and the proposed policy $\policyvec^\toe$, respectively.
Assume that $\beta_i^{(m)}$ is upper bounded and define 
$$\beta_{max} = \max_{1\leq i\leq N, 1\leq m\leq M} \beta_i^\um.$$
We can prove that $h^\toe(M)$ is close to $h^*(M)$.

\begin{theorem}\label{theorem:optimality}
    For the proposed policy $\policyvec^\toe$, we have
    \begin{align}
        0 & \leq h^*(M) - h^\toe(M) \nonumber \\
        & \leq \frac{1}{b} + \max_{1\leq i\leq N} \frac{2\sqrt{\beta_{max}}} {\sqrt{\sum_{m \in \mathcal{C}_i} \left(s_i^\um+1/\beta_i^\um\right)^{-1}}}
        \nonumber
    \end{align} 
\end{theorem}

The proof of this theorem is presented in Appendix~\ref{sec:proof_theorem_optimality}.
Theorem~\ref{theorem:optimality} indicates 
\begin{align*}
    \lim_{M\to +\infty} h^*(M) - h^\toe(M) = 1/b,
\end{align*} 
if for any $1\leq i \leq N$
\begin{align}\label{eq:opt_condition}
    \lim_{M\to +\infty} \sum_{m \in \mathcal{C}_i} \left(s_i^\um+1/\beta_i^\um\right)^{-1} \to +\infty.
\end{align}
Condition (\ref{eq:opt_condition}) states that as the user population grows, 
the overall request rate for $d_i$ from the user set $\mathcal{C}_i$ also keeps increasing.
Under such conditions, we would have more and more overhearing opportunities 
as the user population increases.
Theorem~\ref{theorem:optimality} reveals an insight that 
we could achieve near optimal caching performance 
by strategically leveraging the overhearing opportunities (i.e., $\policyvec^\toe$),
if the demand of each data item increases consistently  
when it is exposed to a larger user population.

Different from the policy $\policyvec^\tcoe$ that can achieve 
asymptotically optimal performance for homogeneous demand settings,
the policy $\policyvec^\toe$ designed for heterogeneous settings 
always has a $1/b$ hit ratio gap with the optimal policy. 
The reason is that $\policyvec^\toe$ adopts an overhearing only mechanism 
and restricts the overhearing TTL $\omega_i^\um$
to take either value $s_i^\um$ or $+\infty$ for the ease of analysis. 
Considering that the cache size $b$ is typically large, 
$\policyvec^\toe$ achieves reasonably good performance.

\section{Evaluation}
\label{sec:eva}
In this section, we will validate the theoretical results by 
evaluating the empirical performance of 
the proposed $\policyvec^\tcot$ and $\policyvec^\tcoe$ policies 
and comparing them with the following benchmarks:
\begin{itemize}
    \item The optimal overhearing-only policy: this policy is the solution of problem~(\ref{eq:opt_TTL_1})
    with an additional constraint $\tau_i = 0$, $1\leq i\leq N$,
    and can be easily solved using the same approach that solves $\policyvec^\tcot$ and $\policyvec^\tcoe$.
    The optimal overhearing-only policy evicts the data item immediately after a request for it.  
    \item The optimal caching-only policy: this policy is the solution of problem~(\ref{eq:opt_TTL_1})
    with an additional constraint $\omega_i = +\infty$, $1\leq i\leq N$,
    and can be easily solved using standard convex optimization tools.
    It turns out that the optimal caching-only policy caches data items with the largest long term popularities.
    and achieves better performance than various conventional policies (e.g., LRU, LFU) under the setting of this paper. 
    \item LFU policy: the policy caches the most frequently used data items, which is an approximation of the optimal caching-only policy.
    \item LRU policy: the policy caches the most recently used data items.
\end{itemize}
In Experiments~1 and 2, we evaluate the performance of these policies under time-driven 
and event-driven overhearing settings, respectively. 

\textbf{Experiment 1:}
In this experiment, we consider the time-driven overhearing setting. 
{\color{black}Set $b=50$, $N=1000$, $\beta_i = c\cdot i^{-0.8}$ 
with $c = 1/\sum_{i=1}^Ni^{-0.8}$ 
and $s_i = 1/\beta_i$.}
Let $\lambda_i = \gamma \cdot \beta_i$,
where $\gamma = \sum_{i=1}^N \lambda_i$ is the overall overhearing rate. 
\begin{figure}[t]
    \centering
    \begin{subfigure}{.9\linewidth}
        \centering
        \includegraphics[width=6.8cm]{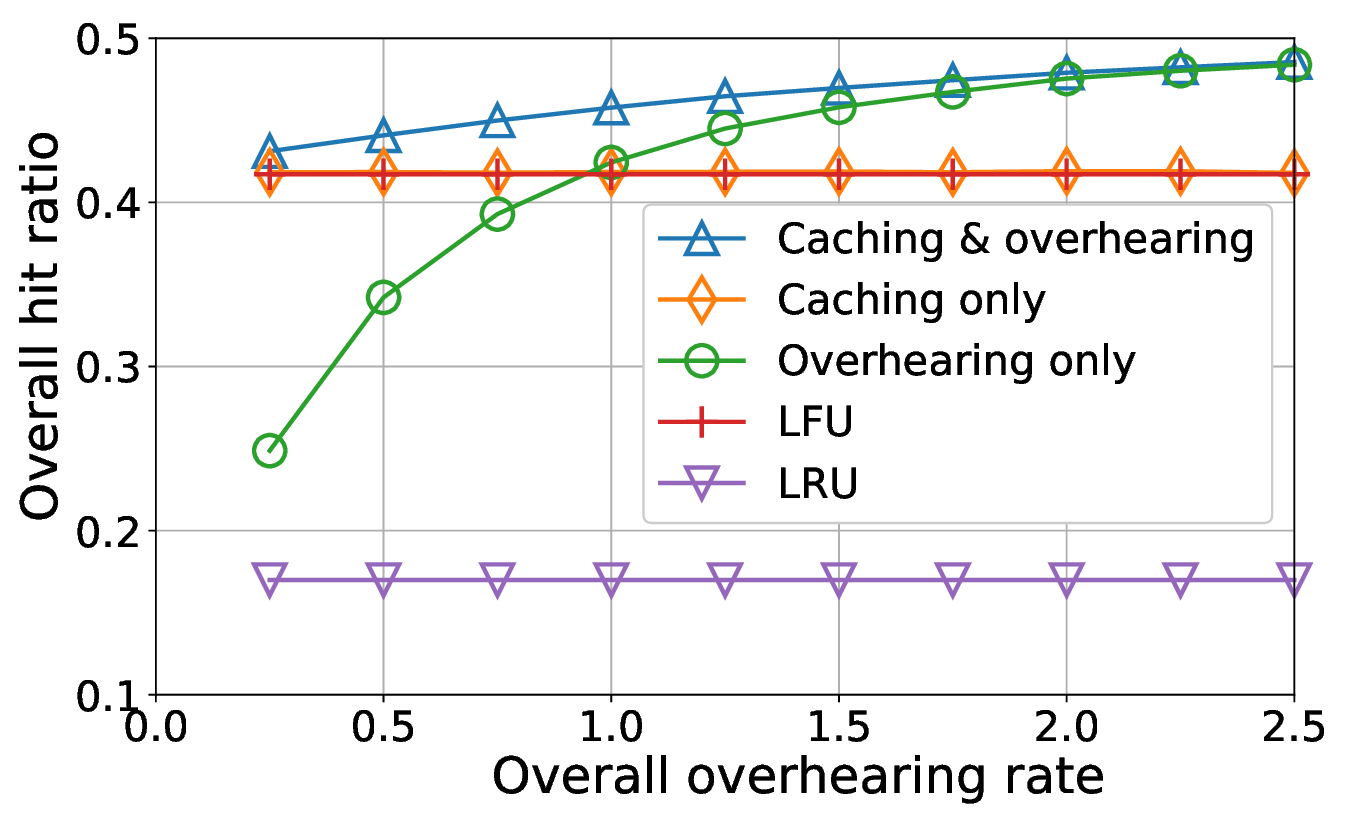}
        \vspace{-2mm}
        \caption{}
        \label{fig:exp1_sub1}
    \end{subfigure}%
    \\
    \begin{subfigure}{.9\linewidth}
        \centering
        \includegraphics[width=6.8cm]{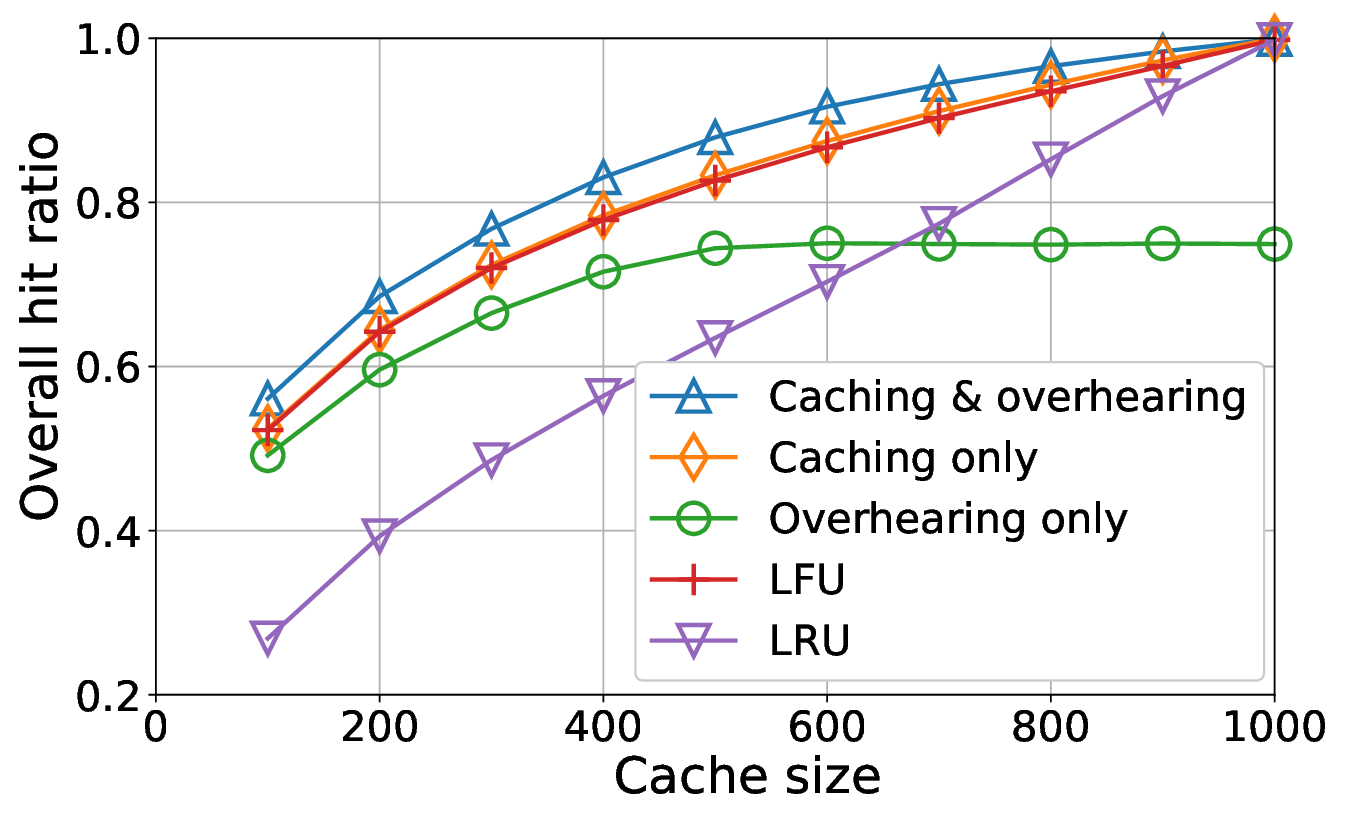}
        \vspace{-2mm}
        \caption{}
        \label{fig:exp1_sub2}
    \end{subfigure}
    \vspace{-2mm}
    \caption{\color{black}Overall hit ratio with time-driven overhearing.}
    \label{fig:exp1}
    \vspace{-5mm}
\end{figure}
Since the number of caches $M$ does not impact the performance for time-driven 
overhearing, we simply set $M=1$.
We evaluate the overall hit ratios under different $\gamma$ values
and depict the results in Figure~\ref{fig:exp1_sub1}.
It can be observed that the proposed optimal caching and overhearing policy $\policyvec^\tcot$
always outperforms the other benchmarks. 
The overhearing-only policy achieves similar performance as $\policyvec^\tcot$ when $\gamma$ is large,
which validates that when there are sufficient overhearing opportunities, 
the overhearing-only policy can achieve near optimal performance. 
However, when $\gamma$ is small, the overhearing-only suffers a lot. 
The overall hit ratios achieved by the caching-only policy, LFU and LRU are constant, 
since they are independent of the overhearing process. 
LFU achieves similar performance as the caching-only policy.
However, LRU achieves much worse performance, since under individual demand dynamics
the most recently used data item may not be popular in the near future. 

Next, we fix $\gamma = 1$ and change the cache size $b$.
The results are plotted in Fig.~\ref{fig:exp1_sub2}.
The caching and overhearing policy $\policyvec^\tcot$ still outperforms the other benchmarks. 
As for the overhearing-only policy, the hit ratio 
will be a constant less than $1$, when the cache size is larger than a threshold.
The reason is that the overall overhearing rate is too low,
and cache is not full even when we maximize the overhearing utilization 
(i.e., set $\omega_i = 0$).  
As a result, further increasing the cache size will not lead to a higher hit ratio.
The caching-only policy, LFU and LRU can achieve a hit ratio $1$,
when the cache is large enough to store all items (i.e., $b=N=1000$).

\textbf{Experiment 2:}
In this experiment, we simulate the event-driven overhearing.
Consider $N=1000$ data items with $\beta_i = c\cdot i^{-0.8}$, $c = 1/\sum_{i=1}^N i^{-0.8}$
and $s_i = 1/\beta_i$ for $1\leq i \leq N$.
\begin{figure}[t]
    \centering
    \begin{subfigure}{.9\linewidth}
        \centering
        \includegraphics[width=6.8cm]{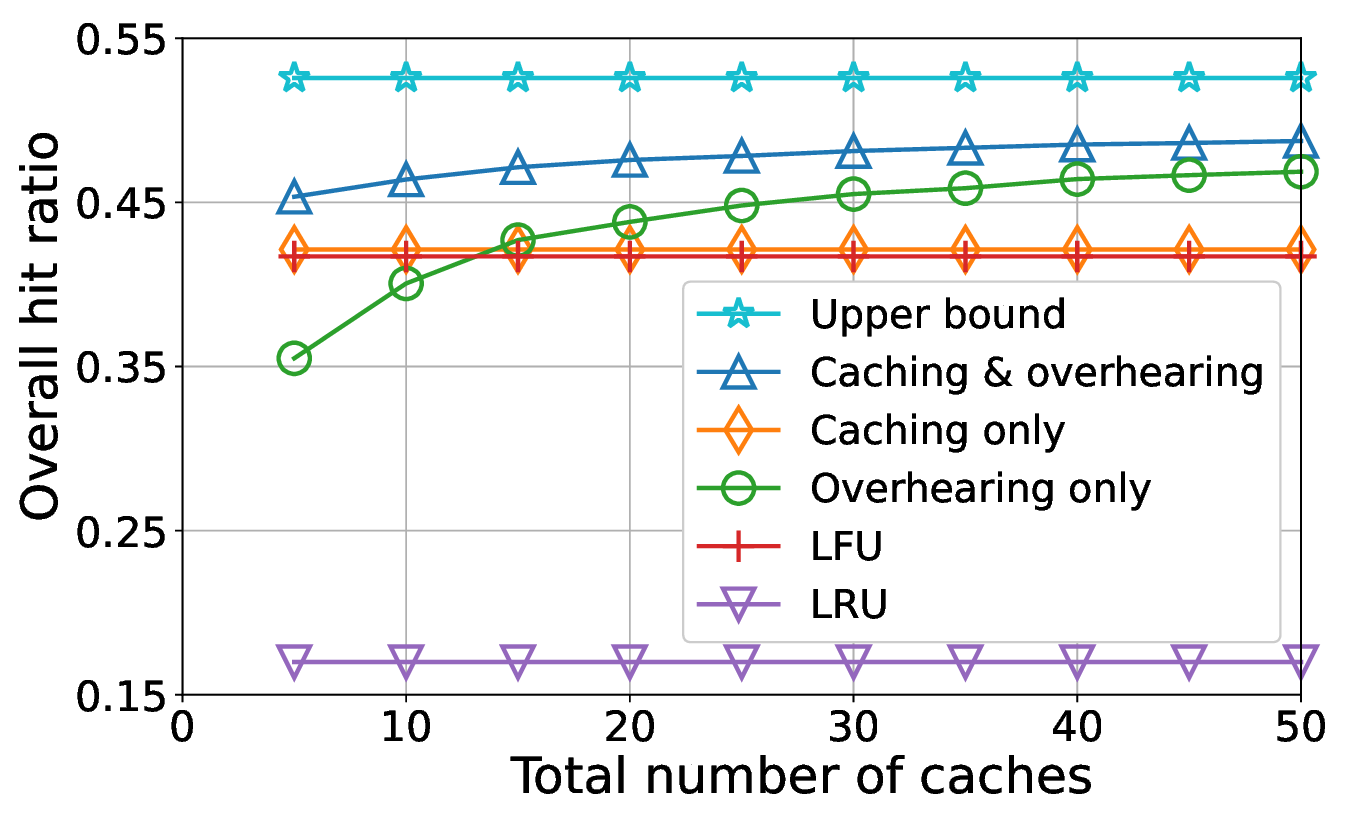}
        \vspace{-2mm}
        \caption{}
        \label{fig:exp2_sub1}
    \end{subfigure}%
    \\
    \begin{subfigure}{.9\linewidth}
        \centering
        \includegraphics[width=6.8cm]{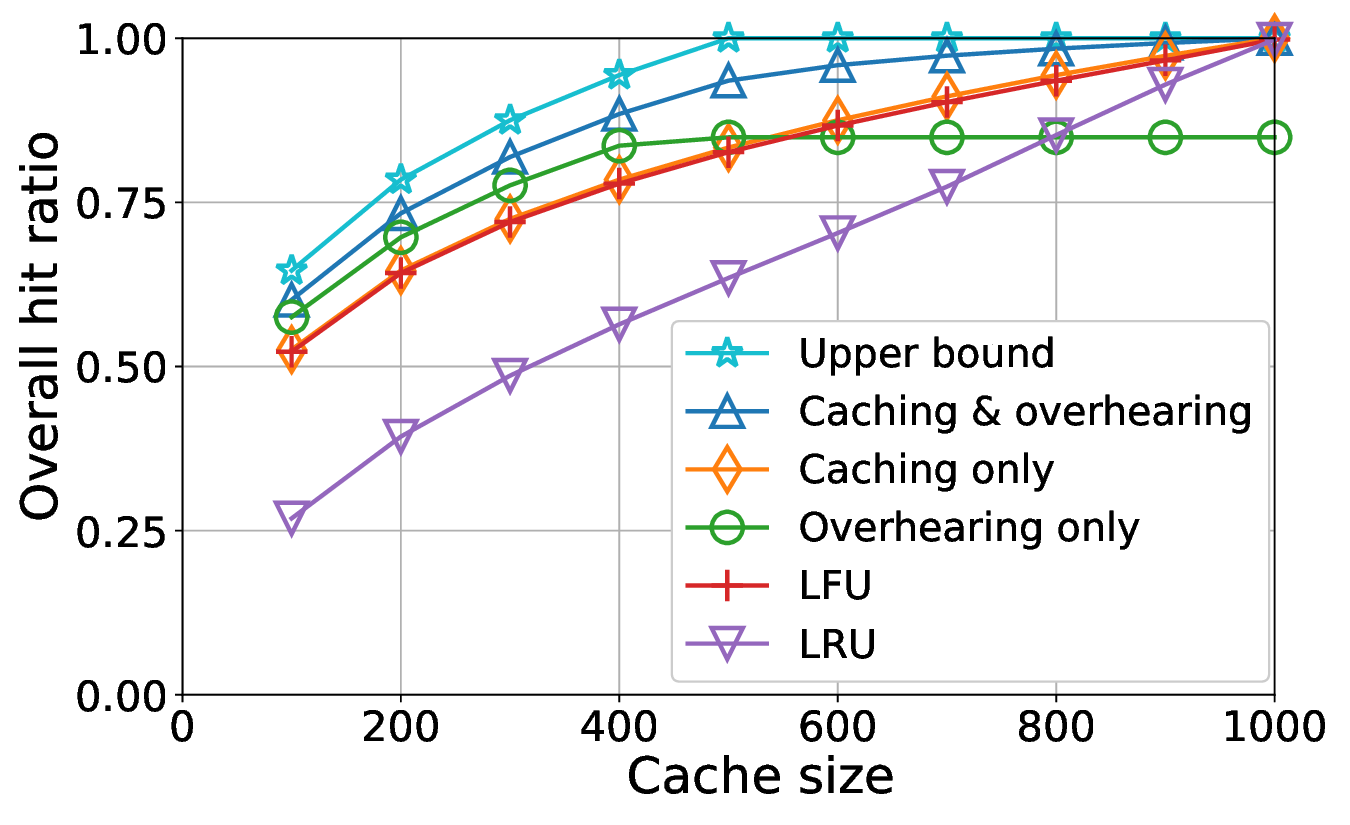}
        \vspace{-2mm}
        \caption{}
        \label{fig:exp2_sub2}
    \end{subfigure}
    \vspace{-2mm}
    \caption{\color{black}Overall hit ratio with event-driven overhearing.}
    \label{fig:exp2}
    \vspace{-5mm}
\end{figure}
We evaluate the overall hit ratios achieved by the proposed overhearing and caching policy $\policyvec^\tcoe$
as well as other benchmarks,
and compare them with the upper bound of the optimal hit ratio derived in Equation~(\ref{eq:upper_hit}).
Note that the policy $\policyvec^\tcoe$ is solved 
using an estimation phase with duration $10000$ 
based on discussions in Section~\ref{sec:discuss}.

{\color{black}First, we set the cache size $b = 50$ and evaluate the hit ratios for different numbers of caches.
The results are plotted in Figure~\ref{fig:exp2_sub1}. }
The proposed caching and overhearing-only policy always achieves the best performance.
When $M$ is small, the overhearing-only policy achieves a much lower hit ratio 
than the caching-only policy. 
When $M$ is large, the overhearing-only policy outperforms the caching-only policy. 
In addition, the hit ratios achieved by
the caching and overhearing policy
will get closer and closer to the upper bound $h^\text{upper}$ as $M$ increases, 
which validates the asymptotic optimality.
LFU achieves similar performance to the optimal caching-only policy
and LRU achieves the worst performance. 

Next, we set $M=50$ and evaluate the hit ratios for different cache sizes. 
The results are presented in Figure~\ref{fig:exp2_sub2}.
The proposed caching and overhearing-only policy $\policyvec^\tcoe$
always achieves the best performance.
When the cache size is less than $500$,
the overhearing-only policy outperforms the caching-only policy, 
because $M$ is relatively large to generate sufficient overhearing opportunities. 
However, for $b>500$, when we further increase the cache size,
the overhearing-only policy cannot achieve a larger hit ratio. 
At this time, the overhearing opportunities become the bottleneck of the system,
and the cache space cannot be fully utilized due to the lack of overhearing opportunities.
In contrast, the proposed caching and overhearing policy, the caching-only policy, LFU and LRU can always benefit from a larger
cache size and achieve a higher hit ratio.   

\section{Conclusion}
\label{sec:conclusion}
Edge caching typically serves a very small group of users with 
individualized data demand.
Hence, caching schemes for an edge need to be substantially 
different from those at the core
that serves a large population of users. 
In this paper, we developed new caching policies optimized for
individualized data demand at the wireless edges.
With the objective to maximize the overall hit ratio, 
we proposed to actively evict the data items that are not 
likely to be requested in the near future and 
bring them back into the cache through overhearing 
when they become popular again. 
In particular, when the overhearing opportunities are time-driven,
the optimization problem turns out to be non-convex.
Nevertheless, by exploiting an informative structure of the optimal solution, 
we converted the original problem to a convex one 
and found the optimal policy $\policyvec^\tcot$. 
When the overhearing opportunities are event-driven, 
the overhearing processes become intractable. 
Still, inspired by the optimality structure 
of the time-driving overhearing setting, 
we proposed a caching and overhearing policy $\policyvec^\tcoe$
which is asymptotically optimal 
as the total number of caches increases.
Both theoretical and numerical results verified that the caching policies designed specifically for edges
could substantially improve the caching efficiency and outperform
the policies designed for the core. 



{\appendices
\section{Proof of Theorem~\ref{theorem:hit_ratio}}
\label{sec:pf_hit_ratio}
Let $\pi^\tc(\tau_i) = \pi^\tco(\tau_i,+\infty)$
denote a caching-only policy that never overhears.
Let $h_i^\tc(\tau_i)$ and $r_i^\tc(\tau_i)$ denote the 
expected hit ratio and cache occupancy of the data item $d_i$
when it is served by the policy $\pi^\tc(\tau_i)$.
Similarly, we can define $\pi^\too(\omega_i) = \pi^\tco(0,\omega_i)$,
$h_i^\too(\omega_i)$ and $h_i^\too(\omega_i)$.

\begin{lemma}\label{lemma:hit_ratio_sepa}
    The expected hit ratio and cache occupancy achieved by a 
    policy $\pi^\tco(\tau_i,\omega_i)$ with $\tau_i\leq\omega_i$
    can be calculated by 
    \begin{align} 
        & h^\tco(\tau_i,\omega_i) = h_i^\tc(\tau_i) + h_i^\too(\omega_i), \nonumber \\
        & r^\tco(\tau_i,\omega_i) = r_i^\tc(\tau_i) + r_i^\too(\omega_i). \nonumber
    \end{align}
\end{lemma}

\begin{proof}[Proof of Lemma \ref{lemma:hit_ratio_sepa}]
Assume there is a request for the data item $d_i$ at time 0
and the next request for $d_i$ will arrive at time $\sigma>s_i$.
We will first analyze the probability of having a cache hit of 
$d_i$ at time $\sigma$, as well as the 
expected time when $d_i$ is stored in the cache in the time 
interval $[0,\sigma]$.

\begin{align}
    \P&[\text{cache hit at time $\sigma$}
        |\text{ $d_i$ is served by $\pi^\tco(\tau_i,\omega_i)$}] \nonumber \\
    & = \P[\sigma > s_i + \omega_i \text{ and $d_i$ is overheard
      during $[\omega_i,\sigma]$}] \nonumber \\
    & \;\;\;\; + \P[\sigma \leq s_i + \tau_i] \nonumber \\
    \P&[\text{cache hit at time $\sigma$}
        |\text{ $d_i$ is served by $\pi^\tc(\tau_i)$}] \nonumber \\
    & = \P[\sigma \leq s_i + \tau_i] \nonumber \\
    \P&[\text{cache hit at time $\sigma$}
        |\text{ $d_i$ is served by $\pi^\too(\omega_i)$}] \nonumber \\
    & = \P[\sigma > s_i + \omega_i \text{ and $d_i$ is overheard
      during $[\omega_i,\sigma]$}]. \nonumber
\end{align}
Hence, we have 
\begin{align}
    \P&[\text{cache hit at time $\sigma$}
        |\text{ $d_i$ is served by $\pi^\tco(\tau_i,\omega_i)$}] 
        \nonumber \\
    & = \P [\text{cache hit at time $\sigma$}
    |\text{ $d_i$ is served by $\pi^\tc(\tau_i)$}] \nonumber \\
    & \;\;\;\; + \P [\text{cache hit at time $\sigma$}
    |\text{ $d_i$ is served by $\pi^\too(\omega_i)$}], \nonumber
\end{align}
which indicates 
$h^\tco(\tau_i,\omega_i) = h_i^\tc(\tau_i) + h_i^\too(\omega_i)$,
since the policy is renewed after each data request. 

Let $T^\tco$, $T^\tc$ and $T^\too$ denote the amount of time 
when $d_i$ is stored in the cache during $[0,\sigma]$,
if $\pi^\tco(\tau_i,\omega_i)$, $\pi^\tc(\tau_i)$
and $\pi^\too(\omega_i)$ are applied, respectively. 
We have $T^\tco = T^\tc + T^\too$,
which indicates 
$r^\tco(\tau_i,\omega_i) = r_i^\tc(\tau_i) + r_i^\too(\omega_i)$.
\end{proof}

Now we are ready to prove Theorem~\ref{theorem:hit_ratio}.
\begin{proof}[Proof of Theorem \ref{theorem:hit_ratio}]
    In order to prove Theorem \ref{theorem:hit_ratio}, 
    we will derive the expected hit ratio and cache occupancy
    achieved by $\pi^\tc(\tau_i), \pi^\too(\omega_i)$ in different 
    parameter regions. 

    \noindent
    \textbf{Case 1:} $\tau_i \leq \omega_i \leq s_i$
    
    \noindent 
    In this case, the cached data item $d_i$ is evicted during the OFF period. 
    Thus, the caching-only policy $\pi^\tc(\tau_i)$ always achieves a $0$ hit ratio 
    and cache occupancy.
    For the overhearing-only policy $\pi^\too(\omega_i)$, 
    without loss of generality, we assume that 
    the most recent request arrives at time $0$.
    We will analyze the probability that the next request for $d_i$ is a hit,
    and its expected cache occupancy. 
    Let $X_i$ denote the time when the next request for $d_i$ arrives, 
    and $Y_i$ denote the time when we overhear $d_i$ for the first time
    after the deaf period (i.e., after time $\omega_i$).
    We have 
    \begin{align*}
        h_i^\too(\omega_i) & = \P[\text{The next request for $d_i$ is a hit 
            under $\pi^\too(\omega_i)$}] \\
        & = \P[Y_i + \omega_i\leq X_i] \\
        & = \P[Y_i + \omega_i \leq s_i] 
            + \P[s_i < Y_i + \omega_i \leq X_i] \\
        & = 1 - \exp\left(-\lambda_i(s_i-\omega_i) \right) \nonumber \\
        & \;\;\;\; + \exp\left(-\lambda_i(s_i-\omega_i) \right) 
            \cdot \frac{\lambda_i}{\lambda_i+\beta_i} \\
        & = 1 - 
         \exp\left(-\lambda_i(s_i-\omega_i) \right) 
        \cdot \frac{\beta}{\lambda_i+\beta_i}
    \end{align*}
    The cache occupancy of $\pi^\too(\omega_i)$ is 
    \begin{align*}
        r_i^\too(\omega_i) 
        & = \frac{1}{\expect[X_i]}
        \left( 
            \P[Y_i+\omega_i<s_i] \right. \\
        & \hspace{18mm} \cdot \expect[X_i-Y_i-\omega_i|Y_i+\omega_i<s_i] \\
        & \hspace{16mm} 
            + \P[s_i < Y_i + \omega_i \leq X_i] \nonumber \\
        & \hspace{18mm}\left. \cdot 
            \expect[X_i-Y_i-\omega_i|s_i < Y_i + \omega_i \leq X_i]
        \right)\\
        & = \frac{1}{\expect[X_i]} 
        \biggl(
            (1-\exp(-\lambda_i(s_i-\omega_i))) \nonumber \\
        & \hspace{16mm} \cdot
            \left(\frac{s_i-\omega_i}{1-\exp(-\lambda_i(s_i-\omega_i))} 
            -\frac{1}{\lambda_i} + \frac{1}{\beta_i}\right)  \\
        & \hspace{16mm}
            + \exp(-\lambda_i(s_i-\omega_i)) \frac{\lambda_i}{\lambda_i+\beta_i} 
            \frac{1}{\beta_i}
        \biggr) \\
        & = \frac{1}{\expect[X_i]} 
        \left(
            s_i-\omega_i + \exp(-\lambda_i(s_i-\omega_i))
            \frac{\beta_i}{\lambda_i(\lambda_i+\beta_i)} \right. \nonumber \\
        & \hspace{16mm} \left.
            - \frac{1}{\lambda_i} + \frac{1}{\beta_i} 
        \right).
    \end{align*}

    \noindent 
    \textbf{Case 2: $\tau_i \leq s_i \leq \omega_i$}

    \noindent
    In this case, the expected hit ratio and cache occupancy of the 
    caching-only policy are all $0$ similar to Case 1.
    For the overhearing-only policy, we have 
    \begin{align*}
        h_i^\too(\omega_i) & = \P[\text{The next request for $d_i$ is a hit 
            under $\pi^\too(\omega_i)$}] \\
        & = \P[s_i < Y_i + \omega_i \leq X_i] \\
        & = \exp\left(-\lambda_i(s_i-\omega_i) \right) 
            \cdot \frac{\lambda_i}{\lambda_i+\beta_i},
    \end{align*}
    and
    \begin{align*}
        r_i^\too(\omega_i) 
        = & \frac{1}{\expect[X_i]}
            \P[s_i < Y_i + \omega_i \leq X_i] \\ 
        & \cdot
            \expect[X_i-Y_i-\omega_i|s_i < Y_i + \omega_i \leq X_i]\\
        = & \frac{1}{\expect[X_i]}  
            \exp(-\lambda_i(s_i-\omega_i)) \frac{\lambda_i}{\lambda_i+\beta_i} 
            \frac{1}{\beta_i}.
    \end{align*}

\noindent
\textbf{Case 3: $s_i \leq \tau_i \leq \omega_i$}

\noindent
In this case, $h_i^\too(\omega_i)$ and $r_i^\too(\omega_i)$ are exactly the 
same as those in Case~2.
As for $\pi^\tc(\tau_i)$, we have
\begin{align*}
    h_i^\tc(\tau_i) & = \P[\text{The next request for $d_i$ is a hit 
            under $\pi^\tc(\tau_i)$}] \\
            & = \P[X_i \leq \tau_i] \\
            & = 1- \exp(-\beta_i(\tau_i-s_i)).
\end{align*}
The expected cache occupancy can be calculated as 
\begin{align*}
    r_i^\tc(\tau_i) & = \frac{1}{\expect[X_i]}
    \left(
        \P[X_i \geq \tau_i] \cdot \tau_i 
        + \P[X_i < \tau_i] \cdot \expect[X_i|X_i < \tau_i]
    \right) \\
    & =  \frac{1}{\expect[X_i]}
    (
        \P[X_i \geq \tau_i] \cdot \tau_i 
        - \P[X_i \geq \tau_i] \cdot \expect[X_i|X_i \geq \tau_i] \\
        & \hspace{4mm} + \P[X_i \geq \tau_i] \cdot \expect[X_i|X_i \geq \tau_i] \\
        & \hspace{4mm} + \P[X_i < \tau_i] \cdot \expect[X_i|X_i < \tau_i]
    )\\
    & = \frac{1}{\expect[X_i]}
    (-\P[X_i \geq \tau_i] \cdot \expect[X_i - \tau_i|X_i \geq \tau_i]
    +\expect[X_i]) \\
    & = \frac{1}{\expect[X_i]} 
    \frac{1}{\beta_i}(1-\exp(-\beta_i(\tau_i-s_i))).
\end{align*}
Then applying Lemma~\ref{lemma:hit_ratio_sepa} completes the proof.  
\end{proof}

\section{Proof of Lemma~\ref{lemma:boundary}}
\label{sec:proof_lemma_boundary}
The proof of Lemma~\ref{lemma:boundary} consists of two steps.
In Step~1, we will show that 
for $\tau_i \leq s_i$, there exists  $\widetilde{\omega}_i$
such that $h_i^\tco(\tau_i,\omega_i) \leq h_i^\too(\widetilde{\omega}_i)$
and $r_i^\tco(\tau_i,\omega_i)= r_i^\too(\widetilde{\omega}_i)$;
for $\tau_i > s_i$, there exists  $\widetilde{\tau}_i$
such that $h_i^\tco(\tau_i,\omega_i)\leq h_i^\tc(\widetilde{\tau}_i)$
and $r_i^\tco(\tau_i,\omega_i)= r_i^\tc(\widetilde{\tau}_i)$.
Next, in Step~2, we will show that the $h^\trco_i(\cdot)$ function
characterizes the upper boundary.

\noindent 
\textbf{Step 1: }
For $\tau_i \leq s_i$, it is easy to observe that 
$\widetilde{\omega}_i = \omega_i$ will satisfy the property. 
By applying Theorem~\ref{theorem:hit_ratio}, we can verify that 
$h_i^\tco(\tau_i,\omega_i) = h_i^\tco(0,\omega_i) = h_i^\too(\widetilde{\omega}_i)$
and $r_i^\tco(\tau_i,\omega_i) = r_i^\tco(0,\omega_i)= r_i^\too(\widetilde{\omega}_i)$.

For $\tau_i > s_i$, 
we may first solve $\widetilde{\tau}_i$ as the unique solution 
to the equation $r_i^\tco(\tau_i,\omega_i)= r_i^\tc(\widetilde{\tau}_i)$.
Then, we will prove that $h_i^\tco(\tau_i,\omega_i) \leq  h_i^\tc(\widetilde{\tau}_i)$.
By applying Lemma~\ref{lemma:hit_ratio_sepa},
it is equivalent to prove 
$h_i^\tc(\tau_i) + h_i^\too(\omega_i)\leq h_i^\tc(\widetilde{\tau}_i)$.

Since $r_i^\tco(\tau_i,\omega_i) = r_i^\tc(\tau_i) + r_i^\too(\omega_i) 
= r_i^\tc(\widetilde{\tau}_i)$,
we have 
\begin{align}
    r_i^\too(\omega_i) 
    & = r_i^\tc(\widetilde{\tau}_i) - r_i^\tc(\tau_i) \nonumber \\
    & = \frac{1}{\beta_is_i + 1} 
    \left( h_i^\tc(\widetilde{\tau}_i) - h_i^\tc(\tau_i) \right), \nonumber
\end{align}
where the second equation holds due to the fact that 
$\widetilde{\tau}_i>\tau_i > s_i$ and the linear relationship between 
$h_i^\tc$ and $r_i^\tc$, which is illustrated in Fig.~\ref{fig:h_r}
and can be easily proved using Theorem~\ref{theorem:hit_ratio}.
Moreover, Theorem~\ref{theorem:hit_ratio} also indicates that 
\begin{align*}
    r_i^\too(\omega_i) \geq \frac{1}{s_i\beta_i+1} h_i^\too(\omega_i).
\end{align*}
Therefore, we have $h_i^\tc(\widetilde{\tau}_i) - h_i^\tc(\tau_i)\geq h_i^\too(\omega_i)$,
which completes Step~1.

\noindent 
\textbf{Step~2:}
First of all, we will show that for any randomized item policy 
$\policy^\trco(\bq_i, \btau_i, \bomega_i)$
with $|\bq_i| = |\btau_i| = |\bomega_i| = n \geq 2$,
there exist $\widetilde{\bq}_i$, 
$\widetilde{\btau}_i$
and $\widetilde{\bomega}_i$ with 
$|\widetilde{\bq}_i| = |\widetilde{\btau}_i| = |\widetilde{\bomega}_i| = 2$,
such that $h_i^\trco(\bq_i, \btau_i, \bomega_i) 
\leq h_i^\trco(\widetilde{\bq}_i, \widetilde{\btau}_i, \widetilde{\bomega}_i)$
and $r_i^\trco(\bq_i, \btau_i, \bomega_i) 
= r_i^\trco(\widetilde{\bq}_i, \widetilde{\btau}_i, \widetilde{\bomega}_i)$.

The item policy $\policy^\trco(\bq_i, \btau_i, \bomega_i)$ is a randomization
of $n$ deterministic item policies $\policy^\tco(\tau_i^{(j)}, \omega_i^{(j)})$,
$1\leq j \leq n$.
As shown in Fig.~\ref{fig:polygon},
we may plot the cache occupancies and the hit ratios achieved by 
the deterministic item policies $\policy^{\tco}(\tau_i^{(j)}, \omega_i^{(j)})$, 
$1\leq j \leq n$, in a two-dimensional Cartesian coordinate system where  
the x-axis represents the cache occupancy and the y-axis represents the hit ratio.
Based on Theorem~\ref{theorem:hit_ratio_random},
the cache occupancy and hit ratio achieved by any 
randomized item policy $\policy^\trco(\bq_i, \btau_i, \bomega_i)$
must be in the convex hull of these $n$ points 
(i.e., a convex polygon).
\begin{figure}[t]
    \centering
    \includegraphics[width=8cm]{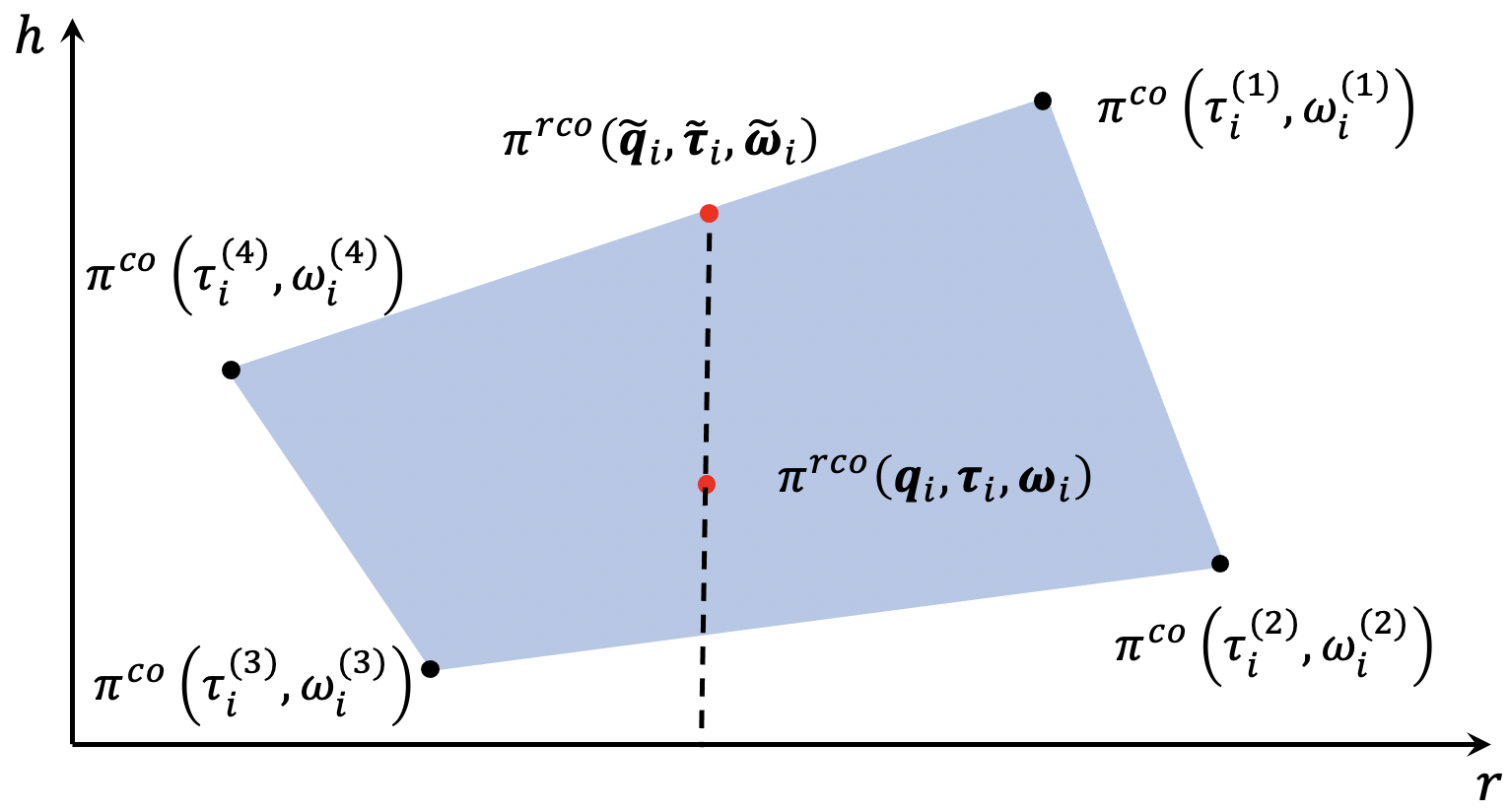}
    \caption{Randomized policies achieving hit ratios
    and cache occupancies in a convex polygon.}
    \vspace{-2mm}
    \label{fig:polygon}
\end{figure}

Next, we may find 
$\policy^\trco(\widetilde{\bq}_i, \widetilde{\btau}_i, \widetilde{\bomega}_i)$
as the item policy  that maximizes
the expected overall hit ratio such that the cache occupancy 
is the same as the one achieved by 
$\policy^\trco(\bq_i, \btau_i, \bomega_i)$.
It is easy to observe that 
$\policy_i^\trco(\widetilde{\bq}_i, \widetilde{\btau}_i, \widetilde{\bomega}_i)$
must be on the 
boundary of the convex polygon and is 
achieved by the randomization of two deterministic policies.
For example, in Fig.~\ref{fig:polygon},
$\policy_i^\trco(\widetilde{\bq}_i, \widetilde{\btau}_i, \widetilde{\bomega}_i)$
is a randomization of 
$\policy^\tco(\tau_i^{(1)},\omega_i^{(1)})$ and 
$\policy^\tco(\tau_i^{(4)},\omega_i^{(4)})$.

Therefore, for any item policy in the set $\policyset^\trco$,
there must exist 
an item policy from the set
\begin{align}
    \widehat{\policyset}^\trco \eqd
    & \Big\{\pi^\trco\left(
        \left(q,1-q\right), 
        \left(\tau^{(1)},\tau^{(2)}\right), 
        \left(\omega^{(1)},\omega^{(2)}\right)
        \right):  \nonumber \\
    & \hspace{20mm} 0\leq q\leq 1,
    0\leq \tau^{(j)} \leq \omega^{(j)}, 
    1\leq j\leq 2 \Big\}. \nonumber
\end{align}
that achieves the same cache occupancy and a higher (or the same)
hit ratio.
An item policy in the set $\widehat{\policyset}^\trco$ is a randomization
of two deterministic item policies 
$\policy^\tco(\tau^{(1)},\omega^{(1)})$ and
$\policy^\tco(\tau^{(2)},\omega^{(2)})$.
We have $\widehat{\policyset}^\trco \subset \policyset^\trco$.
Based on the result of Step~1,
we may know that the any point on the upper boundary of $\mathcal{R}_i^\trco$
must be achieved by 
a randomization of an overhearing-only item policy $\policy^\too(\omega_i)$
and a caching-only item policy $\policy^\tc(\tau_i)$ for some $\omega_i\geq \tau_i \geq 0$.
Based on the result of Step~1,
we can further conclude that for any item policy in the set $\policyset^\trco$,
there must exist 
an item policy from the set
\begin{align}
    & \Big\{\pi^\trco\left(
        \left(q,1-q\right), 
        \left(\tau,0\right), 
        \left(+\infty,\omega\right)
        \right):  \nonumber \\
    & \hspace{12mm} 0\leq q\leq 1,
    0\leq \tau\leq +\infty, 0\leq \omega \leq +\infty\Big\} \nonumber
\end{align}
that achieves the same cache occupancy and a higher (or the same)
hit ratio.
By applying Theorem~\ref{theorem:hit_ratio},
it is easy to show that $h_i^\too(r) \geq h_i^\tc(r)$ for $0\leq r \leq r_i^\too(0)$
and $h_i^\too(r_i^\too(0)) \leq h_i^\tc(r_i^\tc(+\infty)) = 1$.
Therefore, we can prove 
that for any item policy in the set $\policyset^\trco$,
there must exist 
an item policy from the set
\begin{align}
    \widetilde{\policyset}^\trco = & \{\policy^\too(\omega): \omega \geq 0\}  \nonumber \\
    & \cup
    \{\policy^\trco((q,1-q),(+\infty,0),(+\infty,0)): 0 \leq q \leq 1\} \nonumber
\end{align}
that achieves the same cache occupancy and a higher (or the same)
hit ratio.

\section{Proof of Lemma~\ref{lemma:property2}}
\label{sec:proof_lemma_property2}
Consider the overhearing-only item policy $\pi^\too(\omega_i)$ with $\omega_i \geq s_i$
under the event-driven overhearing setting. 
Without loss of generality, we assume that 
the most recent request arrives at time $0$.
We will analyze the probability that the next request for $d_i$ is a hit,
and its expected cache occupancy. 
Let $X_i$ denote the time when the next request for $d_i$ arrives, 
and $Y_i$ denote the time when we overhear $d_i$ for the first time
after the non-overhearing period (i.e., after time $\omega_i$).
We have 
\begin{align*}
    h_i^\too(\omega_i) & = \P[\text{The next request for $d_i$ is a hit 
        under $\pi^\too(\omega_i)$}] \\
    & = \P[Y_i + \omega_i\leq X_i]
\end{align*}
The cache occupancy of $\pi^\too(\omega_i)$ is 
\begin{align*}
    r_i^\too(\omega_i) 
    = & \frac{1}{\expect[X_i]}
        \P[Y_i + \omega_i \leq X_i] \nonumber \\
        & \cdot 
        \expect[X_i-Y_i-\omega_i| Y_i + \omega_i \leq X_i]\\
    = & \frac{1}{\expect[X_i]}
       \P[Y_i + \omega_i \leq X_i] \nonumber \\
       & \cdot 
       \expect[\expect[X_i-Y_i-\omega_i| Y_i + \omega_i \leq X_i]|Y_i]\\
    = & \frac{1}{\expect[X_i]}
       \P[Y_i + \omega_i \leq X_i] \cdot 
       \expect \left[ 1/\beta_i|Y_i \right]\\
    = & \frac{1}{s_i\beta_i + 1} \P[Y_i + \omega_i\leq X_i].
\end{align*}
Therefore, we have $h_i^o(r) = (s_i\beta_i+1)r$ for $\omega_i\geq s_i$,
or equivalently, for $0\leq r \leq r_i^\tco(0, s_i)$.

As for the caching-only item policy $\policy^\tc(\tau_i)$,
the hit ratio and the cache occupancy are the same as 
those in the time-driving scenario,
since $\policy^\tc(\tau_i)$ is independent of the overhearing process.

\section{Proof of Lemma~\ref{lemma:upper}}
\label{sec:proof_lemma_upper}
We want to show that for a given cache size $b$,
any achievable overall hit ratio must be no larger than $h^\text{upper}$. 
Consider an idealized setting, where we can always overhear $d_i$ and store it 
in the cache, $1\leq i\leq N$,
immediately after its OFF period. 
Let $h_i^\text{ideal}(r)$ be the expected hit ratio of $d_i$ 
when the cache occupancy is $r$ under the overhearing only item policy $\policy^\text{\too}(\omega_i)$.
$h_i^\text{ideal}(r)$ is defined for $0\leq r \leq 1/(\beta_is_i+1)$.
It is easy to show that $h_i^\text{ideal}(r) = (\beta_is_i + 1) r$.
Note that a hit ratio $1$ and a cache occupancy $1/(\beta_is_i+1)$ are achieved 
by $\policy^\too(s_i)$.

Furthermore, by applying Theorems~\ref{theorem:hit_ratio} and \ref{theorem:hit_ratio_random},
it is easy to prove that for any $(r, h) \in \mathcal{R}_i^\trco$, we have 
$h \leq h_i^\text{ideal}(r)$.
Therefore, the maximal overall hit ratio of problem~(\ref{eq:opt_TTL_1}) 
should not exceed the maximal overall hit ratio of the following problem
\begin{align}
    \max_{r_i} \hspace{12mm}& \sum_{i=1}^N  p_i \cdot h_i^\text{ideal}(r_i)   
        \hspace{-20mm} & \label{eq:opt_TTL_3} \\
    \text{subject to} \hspace{8mm} 
    & 0 \leq r_i \leq 1, &  \hspace{-20mm} 1\leq i\leq N, \nonumber \\
    &\sum_{i=1}^N r_i \leq b. & \hspace{-0mm} \nonumber
\end{align}
Based on Equation~(\ref{eq:p_i}), we have 
\begin{align*}
p_i \cdot h_i^\text{ideal}(r_i)  = p_i (\beta_is_i+1) r_i =
 \beta_i r_i \left/ \sum_{j=1}^N \frac{1}{s_j+1/\beta_j} \right.
\end{align*}
Since the data items are sorted such that 
$\beta_i$ is decreasing with respect to $i$, 
the optimal solution to problem~(\ref{eq:opt_TTL_3})
is to set $r_i = 1/(s_i\beta_i+1)$ for $1\leq i \leq K$,
$r_{K+1} = b - \sum_{i=j}^K1/(s_j\beta_j+1)$
and $r_i = 0$ for $i> K+1$.
And the maximal overall hit ratio achieved by this optimal solution is 
\begin{align*}
    \sum_{i=1}^K p_i 
        + p_{K+1} \left(\beta_{K+1} s_{K+1} + 1\right)
        \left(b - \sum_{i=1}^{K} \frac{1}{\beta_is_i +1}\right) 
        \eqd h^\text{upper},
\end{align*}
which is an upper bound for the overall hit ratio
that achieved by any feasible solution to problem~(\ref{eq:opt_TTL_1}).
Therefore, we have $h^*(M) \leq h^\text{upper}$ for $\forall M > 0$.

\section{Proof of Theorem~\ref{theorem:asymp_opt}}
\label{sec:proof_theorem_asymp_opt}
Consider $M$ edge caches with
event-driven overhearing opportunities.
Assume that the data item $d_i$ is served by the item policy $\policy^\too(s_i)$.
We define $H_i$ as 
\begin{align}
    H_i \eqd  \lim_{T \to \infty} 
    \frac{\text{Number of hits for $d_i$ during $[0, T]$}}
    {\text{Number of requests for $d_i$ during $[0, T]$}}. \nonumber
\end{align}
In order to Theorem~\ref{theorem:asymp_opt}, 
we will first introduce the following lemma. 
\begin{lemma}\label{lemma:overhear_only}
    Consider $M$ edge caches with
    event-driven overhearing opportunities. 
    If the data item $d_i$ is served by the item policy $\policy^\too(s_i)$,
    then we have
    \begin{align}
    0 \leq 1- H_i 
    \leq
    2\sqrt{\frac{\beta_is_i+1}{M}} \nonumber
    \end{align}
    almost surely.
\end{lemma}
\begin{proof}
    Define $A_i(T)$ 
    as the number of misses of $d_i$ during $[0,T]$,
    and $B_i(T)$ as the number of requests of $d_i$ during $[0,T]$.
    Based on the definition of $H_i$,
    we have $1-H_i = \lim_{T\to +\infty}A_i(T)/B_i(T)$.
    Since the expected inter-request time for $d_i$
    is $s_i+1/\beta_i$.
    the strong law of large numbers implies 
    \begin{align}\label{eq:req_num}
        \lim_{T\to\infty} \frac{B_i(T)}{T} = \frac{M}{s_i+1/\beta_i}
    \end{align}
    almost surely.

    Next, we will analyze the limiting behavior of $A_i(T)/T$.
    We represent each request for $d_i$ by an interval $[t_1, t_2]$,
    if the ON-period of the corresponding request starts at time $t_1$
    and the request arrives at time $t_2$.
    Consider the requests for $d_i$ from all $M$ users. 
    For any given time stamp $t$, we claim that there is at most one 
    miss for $d_i$ among the requests whose
    corresponding intervals intersect with the time stamp $t$.
    Let $[t_1^{(k)}, t_2^{(k)}]$, $1\leq k \leq K$ be the $K$ ($K\leq M$) intervals
    that intersect with the time stamp $t$.
    Without loss of generality, we may assume that these intervals are sorted 
    such that $t_2^{(k)}$ is non-decreasing. 
    If the request represented by the interval $[t_1^{(1)}, t_2^{(1)}]$
    is a miss, then at time $t_2^{(1)}$, the base station will broadcast $d_i$. 
    The other caches will overhear $d_i$ and store it based on the item policy  $\policy^\too(s_i)$.
    Consequently, the requests arriving at $t_2^{(k)}$, $2\leq k\leq K$ will all be hits.
    Using a similar argument, we can conclude that there will be at most 
    one miss among these $K$ requests. 

    In order to derive an upper bound for the number of misses $A_i(T)$,
    we may set multiple time stamps with 
    equal distance apart from  each other, for a given time horizon $T$.
    Let $L$ be the distance between two time stamps and 
    $n = \lfloor T/L\rfloor$.
    Consider $n$ time stamps at time $j\cdot L$, $1\leq j \leq n$.
    We will divide all the requests of $d_i$ into two categories. 
    Category~1 contains the requests that intersect with at least one 
    of these $n$ time stamps. 
    In Category~2, the requests do not intersect with any of 
    the time stamps. 
    Let $A^{(1)}_i(T)$ and  $A^{(2)}_i(T)$ denote 
    the number of requests in Categories~1 and 2, respectively. 
    We have $A_i(T) = A^{(1)}_i(T) + A^{(2)}_i(T)$.
    Next, we will prove upper bounds for $A^{(1)}_i(T)$ and  $A^{(2)}_i(T)$.

    It is easy to observe that $A_i^{(1)}(T)$ can be upper bounded by $n$, 
    since there is at most one miss for the requests intersect with a time stamp. 
    Therefore, we have 
    \begin{align}
        \frac{A_i^{(1)}(T)}{T} 
        \leq \frac{n}{T} \leq \frac{1}{L}. \nonumber
    \end{align}
    Moreover, we can upper bound $A^{(2)}_i(T)$ by the number of requests in 
    Category~2.
    During $[0,T]$, there will be roughly $MT/(s_i + 1/\beta_i)$ requests for 
    the data item $d_i$, where $s_i + 1/\beta_i$ is the expected inter-request 
    time for $d_i$.
    A request belongs to Category~2 only when the length of its ON-period 
    is less than $L$, which occurs with a probability $1-\exp(-\beta_iL)$.
    Therefore, by the strong law of large numbers, 
    we have 
    \begin{align} 
        \lim_{T \to + \infty} \frac{A^{(2)}_i(T)}{T} 
        \leq \frac{M(1-\exp(-\beta_iL))}{s_i + 1/\beta_i}
        \leq \frac{M\beta_iL}{s_i + 1/\beta_i}, \nonumber
    \end{align}
    almost surely. 
    Combining $A^{(1)}_i(T)$ and $A^{(2)}_i(T)$ yields 
    \begin{align}
        \lim_{T \to +\infty} \frac{A_i(T)}{T}
        & = \lim_{T \to +\infty} \frac{1}{T} \left(A^{(1)}_i(T)+A^{(2)}_i(T)\right)
        \nonumber \\
        & \leq \frac{1}{L} + \frac{M\beta_i L}{s_i + 1/\beta_i}
        \nonumber
    \end{align}
    almost surely for any $0< L <+\infty$.
    By selecting $L = ((s_i+1/\beta_i)/(M\beta_i))^{1/2}$, we have
    \begin{align}\label{eq:miss_num}
        \lim_{T \to +\infty} \frac{A_i(T)}{T}
        \leq 2 \sqrt{\frac{M\beta_i}{s_i+1/\beta_i}}
    \end{align}
    almost surely. 

    Thus, combining (\ref{eq:req_num}) and (\ref{eq:miss_num}) implies 
    \begin{align}
        1-H_i
        = \lim_{T \to +\infty} \frac{A_i(T)}{B_i(T)} 
        \leq 2\sqrt{\frac{\beta_is_i + 1}{M}} \nonumber
    \end{align}
    almost surely,
    which completes the proof of the upper bound. 
    In addition, 
    since the number of misses must not exceed the number of requests, 
    $1-H_i$ is lower bounded by zero. 
\end{proof}

Now, we are ready to prove Theorem~\ref{theorem:asymp_opt}.

\begin{proof}[Proof of Theorem~\ref{theorem:asymp_opt}]
First, we define an overhearing only policy as follows.
Let the overhearing only element policy $\policy^\too(s_i)$ serve $d_i$, $1\leq i \leq K$,
where $K$ in defined in (\ref{eq:upper_K}).
Based on the definition of $K$, we must have $\sum_{i=1}^K r_i^\too(s_i) \leq b$,
where $r_i^\too(s_i)$ is the cache occupancy of $d_i$ under $\policy^\too(s_i)$.
We serve $d_{K+1}$ by $\policy^\too(s_{K+1})$, if $\sum_{i=1}^{K+1} r_i^\too(s_i) \leq b$.
Otherwise, we serve $d_{K+1}$ by $\policy^\too(\omega_{K+1})$, such that 
$r_{K+1}^\too(\omega_{K+1}) = 1-\sum_{i=1}^{K} r_i^\too(s_i)$.
We note that 
\begin{itemize}
    \item if $d_{K+1}$ is served by $\policy^\too(\omega_{K+1})$ with $\omega_{K+1} \neq s_{K+1}$
    under the proposed overhearing only policy,
    then based on Lemma~\ref{lemma:property2}, we have $h_i^\too(\omega_{K+1}) > h_i^\text{upper}$, where 
    $h_i^\text{upper}$ is the hit ratio of $d_i$ under the idealized policy defined in Section~\ref{sec:opt_general};
    \item the proposed overhearing only policy cannot outperform $\policyvec^\tcoe$,
    since the overhearing only policy is a feasible solution of problem~(\ref{eq:opt_TTL_4}),
    while $\policyvec^\tcoe$ is the optimal solution.
\end{itemize}

Let $h^\too_i(M) = \expect[H_i]$ denote the expected hit ratio of $d_i$
achieved by the proposed overhearing only policy.
By applying Lemma~\ref{lemma:overhear_only} and the fact that $H_i$ is bounded, 
we have
\begin{align}
    h^\too_i(M) \geq 1- 2\sqrt{\frac{\beta_is_i + 1}{M}}. \nonumber
\end{align}


Under the idealized policy proposed in Section~\ref{sec:opt_general},
the hit ratios for $d_i$, $K+1 < i \leq N$, must be zeros.  
Therefore, we can conclude that
\begin{align} 
    & h^\text{upper} - h^\tcoe(M) \nonumber \\
    & \leq  h^\text{upper} - \sum_{i=1}^N p_i h_i^\too(M) 
    \leq \sum_{i=1}^{K+1} p_i (1-h_i^\too(M)) \nonumber \\
    & \leq \sum_{i=1}^{K+1} 2 p_i 
    \sqrt{\frac{\beta_is_i + 1}{M}} 
    \leq \max_{1\leq i\leq K+1} 2\sqrt{\frac{\beta_is_i + 1}{M}}. \nonumber
\end{align}
\end{proof}

\section{Proof of Theorem~\ref{theorem:optimality}}
\label{sec:proof_theorem_optimality}
Consider $M$ edge caches with
event-driven overhearing opportunities described in Section~\ref{sec:genearlized_event_driven}.
We define $H_i$ as hit ratio of $d_i$ of the entire system
\begin{align}
    H_i \eqd  \lim_{T \to \infty} 
    \frac{\text{Number of hits for $d_i$ during $[0, T]$}}
    {\text{Number of requests for $d_i$ during $[0, T]$}}. \nonumber
\end{align}
Using the same approach that proves Lemma~\ref{lemma:overhear_only},
we can establish the following lemma.
\begin{lemma}\label{lemma:overhear_only_generalized}
    Under the proposed policy $\policyvec^\toe$, we have
    \begin{align}
    0 \leq 1- H_i 
    \leq
    \frac{2\sqrt{\beta_{max}}} {\sqrt{\sum_{m \in \mathcal{C}_i} \left(s_i^\um+1/\beta_i^\um\right)^{-1}}} \nonumber
    \end{align}
    almost surely.
\end{lemma}
We omit the proof of Lemma~\ref{lemma:overhear_only_generalized}, 
since it is an easy extension of the proof of Lemma~\ref{lemma:overhear_only}.

\begin{proof}[Proof of Theorem~\ref{theorem:optimality}]
    Consider an idealized scenario, where we could overhear any data item at any time 
    and receive it without any cost or latency.
    Under this idealized scenario, there is an overhearing only policy where 
    for any cache $m$, each data item $d_i$ is served by an overhearing only item policy $\policy^\too(\omega_i^\um)$.
    We use $h^\text{ideal}(M)$ to denote the maximum overall expected hit ratio of the system under the idealized scenario.
    Let $\hat{h}^\text{ideal}(M)$ to denote the overall expected hit ratio of the system achieved by the proposed $\policyvec^\toe$
    in the idealized scenario.
    
    Note that, in the idealized scenario, the overall cache occupancy of cache $m$ achieved by $\policyvec^\toe$ should be 
    in the range of $(b-1, b]$, due to Equation~(\ref{eq:K_m}).
    Therefore, we have
    \begin{align}\label{eq:h_gap_1}
        0 \leq h^\text{ideal}(M) - \hat{h}^\text{ideal}(M) \leq 1/b, 
    \end{align}
    since $\policyvec^\toe$ prefers to cache items with the largest caching efficiency 
    (i.e., hit ratio per unit cache occupancy).
    
    Using the same approach that proves Theorem~\ref{theorem:asymp_opt}, 
    we can show that 
    \begin{align}\label{eq:h_gap_2}
        & \hat{h}^\text{ideal}(M) - h^\toe(M)  \nonumber \\ 
        & \hspace{4mm}\leq \max_{1\leq i\leq N} \frac{2\sqrt{\beta_{max}}} {\sqrt{\sum_{m \in \mathcal{C}_i} \left(s_i^\um+1/\beta_i^\um\right)^{-1}}}.
    \end{align}
    Combining (\ref{eq:h_gap_1}), (\ref{eq:h_gap_2}) and the fact that $h^*(M) \leq h^\text{ideal}(M)$
    finishes the proof.
\end{proof}

\bibliographystyle{IEEEtran}
\bibliography{mybib}

\vfill

\end{document}